\documentclass{lmcs}
\usepackage[utf8]{inputenc}
\pdfoutput=1

\usepackage{lastpage}
\lmcsdoi{18}{2}{5}
\lmcsheading{}{\pageref{LastPage}}{}{}%
{Mar.~04,~2020}{Apr.~26,~2022}{}

\usepackage{microtype}%if unwanted, comment out or use option "draft"

\usepackage{hyperref}
\theoremstyle{plain} %\crefname{satz}{Satz}{S\"atze}
\usepackage[autostyle]{csquotes}

\usepackage{bm}
\usepackage{amsmath}
\usepackage{url}
\usepackage{amsthm}
\usepackage{color}
\usepackage{graphicx}

\usepackage{amssymb}
\usepackage{cmll}
\usepackage{stmaryrd}
\usepackage{proof}
\usepackage{mathtools}
\usepackage[export]{adjustbox}
\usepackage{xspace}

\usepackage{subcaption}

\theoremstyle{plain}
\newtheorem*{theorem*}{Theorem}
\newcommand{\C}{\mathcal C}
\newcommand{\leftstar}{~{}^*\hspace{-2pt}}

\newcommand{\leftclub}{~{}_{\clubsuit}\hspace{-2pt}}

\newcommand{\mmu}{\mu}
\newcommand{\reddr}{\todo{\redd_{\heartsuit}}}
\newcommand{\redds}{\blue{\redd_{\clubsuit}}}

\newcommand{\st}{\mid}

\newcommand{\Norms}{\wwLim}

\newcommand{\Q}{\mathbf {Q}}

\newcommand{\PPb}{\PP_{\norma}}

\newcommand{\sth}{~\mid~}

\newcommand{\m}{\mathtt m}
\newcommand{\mm}{\mathtt m}

\newcommand{\mr}{\mathtt r}
\newcommand{\mrho}{\mathtt r}
\newcommand{\ms}{\mathtt s}

\newcommand{\tm}{\mathtt t}

\newcommand{\tmt}{\mathtt t}
\newcommand{\tms}{\mathtt s}
\newcommand{\tmu}{\mathtt u}
\newcommand{\tmr}{\mathtt r}

\renewcommand{\implies}{\Rightarrow}
\newcommand{\dsum}{+}%

\newcommand{\PLambda}{\Lambda_\oplus}
\newcommand{\Val}{\mathcal{V}}

\newcommand{\MeanTime}{\texttt{ETime}}
\newcommand{\seq}[1]{\langle#1_n\rangle_{n\in\Nat}}

\renewcommand{\flat}{\texttt{flat}}

\newcommand{\mset}[1]{[#1]}
\newcommand{\iI}{i\in I}

\newcommand{\kK}{k\in K}
\newcommand{\lam}{\lambda}

\newcommand{\midd}{\mid}

\newcommand{\ie}{\emph{i.e.}\xspace}
\newcommand{\eg}{\emph{e.g.}\xspace}
\newcommand{\ih}{\emph{i.h.}\xspace}

\newcommand{\sem}[1]{[\![#1]\!]}
\newcommand{\den}[1]{[\![#1]\!]}

\newcommand{\mdist}[1]{[#1]}

\renewcommand{\iff}{\emph{iff}\xspace}

\usepackage{xcolor}

\newcommand{\todo}[1]{{ \color{red}{#1}}}

\newcommand{\blue}[1]{{\color{blue}{#1}}}

\newcommand{\RED}[1]{{\color{red}{#1}}}

\renewcommand{\AA}{C}

\newcommand{\PP}{\mathbf{P}}
\newcommand{\QQ}{\mathbf{Q}}

\newcommand{\zero}{\bm 0}
\newcommand{\PARS}{\textbf{\textsf{pars}}\xspace}
\newcommand{\pars}{\PARS}
\newcommand{\AST}{\texttt{AST}\xspace}
\newcommand{\WN}{\texttt{WN}\xspace}
\newcommand{\SN}{\texttt{SN}\xspace}
\newcommand{\UN}{\texttt{UN}\xspace}

\newcommand{\atlim}{\mathtt{\infty}}

\newcommand{\WCRlim}{$\mathtt{WCR^{\atlim}}$\xspace}
\newcommand{\UNlim}{$\mathtt{UN^{\atlim}}$\xspace}

\newcommand{\WNlim}{$\mathtt{WN^{\infty}}$\xspace}
\newcommand{\SNlim}{$\mathtt{SN^{\infty}}$\xspace}

\newcommand{\LimP}{$\mathtt{LIM}$\xspace}

\newcommand{\Lim}{\mathtt{Lim}}

\newcommand{\tolim}{\xRightarrow{\infty}}

\newcommand{\pr}{\norma}

\newcommand{\pLim}{\Lim_{\pr}}
\newcommand{\tolimp}{\tolim{}_{\mkern-12mu\pr}}

\newcommand{\toinf}{\red\conv{\obs}}
\newcommand{\tocpo}{\red\conv{\obs}}
\newcommand{\DST}[1]{\mathsf{Dst}(#1)}
\newcommand{\DSTNF}{\mathsf{Dst}(\NFA)}
\newcommand{\DSTF}[1]{\mathsf{Dst^{F}}(#1)}

\newcommand{\MDST}[1]{\mathsf{m}#1}
\newcommand{\MA}{\MDST{A}}

\newcommand{\supp}[1]{\mathtt{Supp}(#1)}

\newcommand{\A}{\mathcal{A}}

\newcommand{\red}{\rightarrow}

\newcommand{\len}{\texttt{length}}

\newcommand{\redd}{\mathrel{\rightrightarrows}}
\newcommand{\Red}{\redd }

\newcommand{\N}{\scalebox{.6}[1.0]{NF}}

\newcommand{\nnorm}[1]{\norm{\nf{#1}}}

\newcommand{\snf}{{\scriptscriptstyle{\textsf {NF}}}}

\newcommand{\nf}[1]{#1^{\nfsym}}
\newcommand{\nfsym}{ \scriptscriptstyle{\textsf {NF}}  }

\newcommand{\Term}[1]{\N_{#1}}

\renewcommand{\S}{\mathcal{S}}
\newcommand{\T}{\mathcal{T}}

\newcommand{\muflat}{{ {\m^{\mathsf{dst}}}}}

\newcommand{\norm}[1]{\|#1\|}

\newcommand{\norma}{\scriptscriptstyle{\|\|  } }  %\scriptscriptstyle{\textsf {NF}}
\newcommand{\rel}[2]{\w{#1} \geq  \w{#2}}
\newcommand{\relnf}[2]{#1 \geq_{\snf} #2}
\newcommand{\relnorm}[2]{#1 \geq_{\norma} #2}

\newcommand{\eqflat}[2]{#1 =_{\flat} #2}
\newcommand{\eqnf}[2]{#1 =_{\snf} #2}
\newcommand{\eqnorm}[2]{ #1 =_{\norma} #2}

\newcommand{\LEB}{local $\obs$-RD\xspace}
\newcommand{\EB}{$\obs$-RD\xspace}

\newcommand{\wRD}{$\obs$-RD\xspace}

\newcommand{\LD}{\obs\mbox{-LB}}
\newcommand{\LB}{\obs\mbox{-LB}}
\newcommand{\B}{\obs\mbox{-better}}

\newcommand{\Nat}{\mathbb{N}}
\newcommand{\Real}{\mathbb{R}}
\newcommand{\Set}{\mathbb{S}}

\newcommand{\true}{{\texttt T}}
\newcommand{\false}{{\texttt F}}

\renewcommand{\obs}{\mathtt{obs}}

\newcommand{\cpo}[1]{\bm #1}

\newcommand{\cpor}{{\cpo a}}
\newcommand{\cpos}{{\cpo b}}
\newcommand{\cpou}{{\cpo c}}

\mathchardef\mhyphen="2D % chktex 18
\newcommand{\xLim}[1]{\Lim_{#1}}
\newcommand{\wLim}{\xLim{\obs}}
\newcommand{\wwLim}{\xLim{\weight}}

\renewcommand{\conv}[1]{\mkern-2mu{}^\infty_{#1}~}

\newcommand{\tolimO}[1]{#1 \conv{\obs}}
\newcommand{\tolimww}[1]{#1 \conv{\weight}}

\newcommand{\ex}{\textsc{s}}

\newcommand{\ered}{\uset{\ex~}{\red}}

\newcommand{\w}[1]{ \obs(#1)}
\newcommand{\weight}{\mathtt {w}}
\newcommand{\ww}[1]{\weight(#1)}

\newcommand{\redr}{\todo{\red_{\heartsuit}}}
\newcommand{\reds}{\blue{\red_{\clubsuit}}}

\newcommand{\NFA}{\N_{\A}}

\newcommand{\refsec}[1]{Section~\ref{sec:#1}}
\newcommand{\reflem}[1]{Lemma~\ref{l:#1}}
\newcommand{\refex}[1]{Ex.~\ref{ex:#1}}
\newcommand{\refthm}[1]{Thm.~\ref{thm:#1}}

\newcommand{\refprop}[1]{Prop.~\ref{prop:#1}}
\newcommand{\refdef}[1]{Def.~\ref{def:#1}}

\newcommand{\head}{\mathsf{h}}
\newcommand{\weak}{\mathsf{w} }

\newcommand{\hred}{\uset{\head}{\red}}
\newcommand{\wred}{\uset{\weak}{\red}}

\makeatletter
\newcommand{\uset}[3][0ex]{%
	\mathrel{\mathop{#3}\limits_{
			\vbox to#1{\kern-6\ex@
				\hbox{$\scriptstyle#2$}\vss}}}}
\makeatother

\usepackage[normalem]{ulem}

\usepackage{forest}
\newcommand{\ropen}[1]{{[#1[}} % chktex 9
\newcommand{\lopen}[1]{{]#1]}} % chtkex 9 chktex 9

\begin{document}

% \title{Probabilistic  Rewriting and Asymptotic Behaviour:\texorpdfstring{\\}{}	  on   Termination and Unique Normal Forms}
\title[Probabilistic  Rewriting and Asymptotic Behaviour]{Probabilistic  Rewriting and Asymptotic Behaviour:\texorpdfstring{\\}{}  on   Termination and Unique Normal Forms}

\author[C.~Faggian]{Claudia Faggian}	%required
\address{IRIF, CNRS, Universit\'e de Paris, F-75013 Paris, France}	%required

%\date{}

\begin{abstract}

	 While a mature body of work supports the study of rewriting systems,  abstract tools for Probabilistic Rewriting are still limited.
	 In this paper we study  the question of  \emph{uniqueness of the result} (unique limit distribution),  and develop a set of  proof techniques	  to analyze and compare \emph{reduction strategies}. The goal is to have  tools to support the \emph{operational} analysis of  \emph{probabilistic} calculi (such as probabilistic lambda-calculi) where evaluation
	 allows for different  reduction choices (hence different reduction paths).

\end{abstract}

\maketitle

\section{Introduction}
\emph{Rewriting Theory}~\cite{Terese03} is a foundational theory of computing. Its
impact extends to both  the theoretical side of computer science,  and  the  development of programming languages. A clear example of both aspects is the paradigmatic term rewriting system, $\lambda$-calculus, which is also the  foundation of functional programming.
\emph{Abstract Rewriting Systems (ARS)} are the general theory which  captures the common substratum of rewriting theory,
independently of the particular structure of the objects. It studies
 properties of terms transformations, such as normalization, termination,  unique normal form,
and the relations among them. Such results are a powerful set of tools which can be used when we study the computational and operational properties  of any calculus or  programming language.
Furthermore, the theory  provides tools to study and compare strategies, which become extremely important when a system \emph{may} have reductions leading to a normal form, but  \emph{not necessarily}. Here we need  to know: is there a  strategy which is guaranteed to lead to a normal form, if any exists (\emph{normalizing} strategies)?
{Which strategies  diverge if at all possible (\emph{perpetual}  strategies)? }

\emph{Probabilistic Computation}  models  uncertainty.
Probabilistic forms of automata~\cite{Rabin63}, Turing
machines~\cite{Santos69}, and the $\lambda$-calculus~\cite{Saheb-Djahromi78} exist since long.
The pervasive role it is assuming in areas as diverse as robotics, machine learning, natural language processing, has stimulated the research on  probabilistic programming languages, including functional languages~\cite{KollerMP97, RamseyP02,ParkPT05} whose development is increasingly active.
 A typical programming language supports at least discrete distributions by  providing  a probabilistic construct which models sampling from a distribution. This is also the most concrete way to endow the $\lambda$-calculus with probabilistic choice~\cite{DiPierroHW05,LagoZ12,EhrhardPT11}.
Within the vast   research on models of probabilistic systems,  we wish to  mention  that
probabilistic rewriting is the explicit base of PMaude~\cite{AghaMS06}, a  language for specifying
probabilistic concurrent systems.

\emph{Probabilistic Rewriting.}
Somehow surprisingly, while a large and mature body of work supports the study of rewriting systems---even infinitary  ones~\cite{DershowitzKP91,KennawayKSV95}---work on the abstract theory of \emph{probabilistic} rewriting systems is still sparse. The notion of \emph{Probabilistic} Abstract Reduction Systems (PARS) has been introduced by Bournez and Kirchner  in~\cite{BournezK02}, and then extended in~\cite{BournezG05} to account for non-determinism. Recent work~\cite{popl,DiazMartinez17,Kirkeby, Avanzini} shows an increased  research interest.
 The key element in \emph{probabilistic} rewriting is that even when
  the probability that a term leads to a normal form is $1$ (\emph{almost sure termination}, \AST), that degree of certitude
   is typically not reached in any finite number of steps, but it appears as a limit. Think of a rewrite  rule (as in Fig.~\ref{fig:AST}) which rewrites  $c$ to either the value $\true$ or $c$, with equal probability $1/2$. We write this as  $c\red\{c^{1/2},\true^{1/2}\}$.
    After $n$ steps, $c$ reduces to $\true$ with probability
   $\frac{1}{2} + \frac{1}{2^2} + \cdots + \frac{1}{2^n}$. % $\sum_{k:1}^n \frac{1}{2^k}$.
   {Only at the limit} this computation terminates with probability $1$.

The most well-developed  literature on PARS  is   concerned with methods to prove almost sure termination, see e.g.~\cite{ BournezG05, FioritiH15, Huang0CG19, Avanzini} (this interest matches the fact that there is a  growing body of methods to establish \AST~\cite{AgrawalC018,FuC19, KaminskiKMO18, McIverMKK18, LagoFR21}).
However, considering rewrite rules subject to probabilities opens  numerous other questions, which   motivate our  investigation.
%, both on the abstract properties  and on the proof techniques.

%We take a global point of view on the evolution of a system, and
We study  a rewrite relation  which describes  the global evolution of a probabilistic system, for example  a probabilistic program $P$.
The \emph{result} of the computation  is  a probability distribution  $\beta$ over all the possible output  of $P$.
%This is coherent  with the point of view that the meaning of a program is the value it evaluates to.
The intuition  (see~\cite{KollerMP97}) is that the program $P$ is executed, and  random choices are made by sampling. This process  defines a distribution $\beta$ over the various outputs that the program can produce. We  write this $P\tolim \beta$.
%{( $\beta$ does not need to have total  measure $1$, because some runs may  diverge).}

What happens if the
\emph{evaluation} of a term $P$ is  \emph{not deterministic}, in the sense that different reduction choices are available? Remember that non-determinism arises naturally in the $\lambda$-calculus, because  a term may have several redexes. This aspect has practical  relevance to programming.  %the fact that is possible to fire different redexes in parallel,
Together with the fact that the result of a terminating computation is unique, (independently from the evaluation choices),
it  is   key to the inherent parallelism of functional programs (see \eg~\cite{Marlow}).

Assume program  $P$ generates  a  distribution over  booleans
$ \{\true^\frac{1}{16},  \false^\frac{15}{16}\}$;
it is
  desirable  that   the distribution which is computed is unique: it \emph{only depends on the ``input'' (the problem)}, not on the way the computational steps are performed.

When assuming non-deterministic evaluation, several  questions on  PARS arise naturally.
For example: (1.) when---and in which sense---is the result unique?  (naively, if $P\tolim \alpha$ and $P\tolim  \beta$, is $\alpha = \beta$?)
	(2.) Do all rewrite sequences from the same term have the same probability to reach a result?
		 (3.) If not, does there exist a strategy to find a result with greatest probability?

Such questions are relevant  to the theory and  to the practice of computing.
We believe that to study them, we can advantageously  adapt  techniques from  Rewrite Theory.
However, we \emph{cannot assume that standard properties of  ARSs hold for PARSs.} The game-changer  is that termination  appears as a \emph{limit}.  In Section~\ref{proof_techniques} we show that a well-known ARSs property,
 Newman's  Lemma, does not hold for PARSs.
 This is not surprising;  indeed, Newman's  Lemma  is known not to hold in general for infinitary rewriting~\cite{Kennaway92,KlopV05}. Still, our  counter-example points out that moving from ARS to PARS is non-trivial.
There are
 two main issues: we need to find the \emph{right formulation} and the \emph{right proof technique}. It seems then  especially important to have a collection of proof methods which apply well to PARS\@.

\paragraph{Content and contributions.}
Probability is concerned with \emph{asymptotic} behaviour:  what
happens not after a finite number $ n $ of steps, but \emph{when $ n $ tends to
infinity}. In this paper we focus  on   the  asymptotic behaviour of rewrite sequences \emph{with respect to normal forms}---normal form being the most standard notion of result in rewriting.
We  study  computational  properties such as (1.), (2.), (3.) above.
We do so with the point of view of  ARSs, aiming for  properties
which \emph{hold  independently} of the specific nature of
the rewritten objects; the purpose is to have   tools which apply to any probabilistic rewriting system.

\subparagraph{PARS\@.}
After  motivating and introducing our  formalism for PARSs (Section~\ref{sec:pars} and~\ref{sec:formalism}),
in Section~\ref{sec:asymptotic} we formalize  the notion of  limit distribution, and of well-defined result.
Since  in a PARS each term   has different possible
reduction sequences (with each sequence leading to a possibly different limit distribution), to each term  is naturally associated \emph{a set of  limit distributions}.
To study when  a PARS  has a well-defined result is the main focus of the paper.

Recall  a  property which is crucial to the computational interpretation   of a system such as the $\lam$-calculus: if a term has a normal
form, it is unique---meaning that the result of the computation is \emph{well-defined}.
With this in mind, we
investigate in the probabilistic setting an analogue of the ARS notions  of  \emph{Unique Normal Form (\UN)}, and the possibility or necessity to reach a result:  \emph{Normalization (\WN),  Termination (\SN)}.  We  provide methods and {criteria} to establish these properties, and we {uncover  relations} between them.
Specific contributions  are the following.

\begin{itemize}
	\item We  propose  an analogue of \UN\  for PARS\@.   The question was already studied in~\cite{DiazMartinez17} for PARS which are almost surely terminating, but the solution there does not extend to the  general case.

	\item  We investigate the classical ARS method to prove \UN\ via   \emph{confluence}; we uncover that  subtle aspects appear when dealing with a notion
	of result as  a limit. We do prove   an analogue of ``confluence implies \UN'' for PARS---however  the proof is not simply an adaptation of the standard techniques, due to the fact that  the set of limit distributions is---in general---infinite, and it is not guaranteed to have maximal elements (think of $ \ropen{0,1} $ which has a sup, but not a max). % We need to exploit the structure of $\Real$.

\end{itemize}

\subparagraph{Asymptotic rewriting: QARS\@.}
%We observed that the notion of result as a limit does not belong to ARS\@.
To better understand the \emph{asymptotic} behaviour of computation,  in Section~\ref{sec:QARS} we introduce the  setting of  \emph{Quantitative Abstract Rewrite System} (QARS). While motivated from the analysis of probabilistic rewriting, QARSs
abstract  from the probabilistic structure. This allows us
 to capture   the essence of the  arguments, and to
 separate the properties which really depend on  probability (and its specific properties) from those which are only concerned with the fact that  results are limits.

QARS are  a natural refinement of the notion of Abstract Rewrite Systems with Information content (ARSI), introduced by Ariola and Blom~\cite{AriolaBlom02}. There, to the ARS   is
associated  a partial order that expresses the \emph{information content} of the elements. We adopt the same view.
ARSI however have a notion of limit which is  tailored   to  infinite normal forms in the sense of   B{\"o}hm trees~\cite{Barendregt} and Levy-Longo trees~\cite{LevyPhD}. With QARS, we simply move from partial orders (and a specific definition of limit), to \emph{$\omega$-complete partial orders}---this is enough  to capture also probabilistic  computation. %It also open some new questions, which we deal with in \RED{Section~\ref{}}.

First, we study the properties of  limits.
Then, we   provide \emph{a  set of proof techniques} to support the \emph{asymptotic}  analysis of  \emph{reduction strategies}.
To do so, we  extend to our setting a method which was introduced for ARSs  by  Van Oostrom~\cite{Oostrom07}, and which   is  based on Newman's property of Random Descent (RD)~\cite{Newman,Oostrom07,OostromT16}  (see Section~\ref{background}). The Random Descent method turns out to be well-suited to asymptotic and probabilistic rewriting, providing
a  useful \emph{family of tools}.
In analogy to their counterpart in~\cite{Oostrom07}, we   generalize in a quantitative way the  notions of Random Descent  (which becomes \EB) and of  being \emph{better}  (which become $ \B $);
both properties are here parametric with respect to the information content which we wish to observe.

A significant   technical feature (inherited from~\cite{Oostrom07}) is that both notions of \EB\ and $\B$ come with a characterization via a {\emph{local
		condition}}, in the sense that only single steps from an object---rather than all possible  sequences of steps---need to be examined.

\subparagraph{Probabilistic rewriting: tools and applications.}
In Sections~\ref{sec:RD_pars} and~\ref{sec:comparing_pars} we specialize the Random Descent techniques to PARS\@.
\begin{itemize}
	\item  \wRD  entails that  all rewrite sequences from a term lead to the \emph{same result}, in the \emph{same expected number of steps} (the average of number of steps, weighted w.r.t.\ probability).
	\item $ \B $   offers   a method to compare strategies (``strategy $\S$ is always better than strategy $\T$'')
	w.r.t.\ the \emph{probability} of reaching a result and the \emph{expected time} to reach a result.
	It provides a sufficient criterion to establish that a strategy is \emph{normalizing} (resp. \emph{perpetual}) \ie\ the strategy is guaranteed to lead to a result with maximal (resp.\ minimal) probability.

\end{itemize}

\noindent
To illustrate their use, we apply these methods to     a   probabilistic $\lambda$-calculus---Weak  Call-by-Value $\lambda$-calculus---which is discussed in Section~\ref{sec:weak}.
A larger  example of application   to probabilistic $\lambda$-calculi  is~\cite{FaggianRonchi}, whose developments  rely also  on the abstract results presented here;   we illustrate  this   in  \refsec{large}.

\begin{rem}[On the term \emph{Random} Descent] Please note that in~\cite{Newman}, the term \emph{Random} refers to  non-determinism (in the choice of the redex), \emph{not to  randomized}  choice.
\end{rem}

\paragraph{Journal vs conference version.} This paper is the journal version of~\cite{Faggian19}. The content has been considerably  extended. In particular, we develop the setting of  QARS (Section~(\ref{sec:QARS})), which  formalizes the notion of asymptotic rewriting, and does not appear in~\cite{Faggian19}. This allows us to
separate the properties which really depend on  probability from those which are concerned with results as limits, cleaning the arguments from unnecessary structure. The study of limits in both probabilistic and non-probabilistic setting is unified to a more general theory.
The results obtained for QARS can be transferred to ARS and PARS alike, but also to other frameworks where reduction is asymptotic.

\subsection{Motivations and Background}
\subsubsection{Probabilistic $\lambda$-calculus, non-deterministic evaluation,\\ and (non-)Unique Result}\label{sec:motivation} % chktex 36
Rewrite theory provides numerous tools to study uniqueness of normal forms, as well as techniques to
study and compare strategies.
This is not the case in the probabilistic setting. Perhaps a reason is that
when extending the $\lambda$-calculus with a choice operator, confluence is lost, as was observed early~\cite{deLiguoroP95};
we illustrate  it in Example~\ref{ex:motivation} and~\ref{ex:motivation2}, which is adapted from~\cite{deLiguoroP95, LagoZ12}.
The {way to deal with this issue}  in probabilistic $\lambda$-calculi  (e.g.~\cite{DiPierroHW05,LagoZ12, EhrhardPT11}) has been to  \emph{fix
	a deterministic reduction  strategy}, typically ``leftmost-outermost''.
To fix a deterministic strategy  is not satisfactory,
neither for the   theory nor the   practice of computing. To understand why this matters, recall for example that confluence of the $\lambda$-calculus is what makes
functional programs  inherently parallel:
every sub-expression can be evaluated in parallel,
still, we can reason on a program using a deterministic sequential model,
because the result of the computation is independent of  the evaluation order (we refer to~\cite{Marlow}, and to  Harper's text ``Parallelism is not Concurrency'' for discussion on  \emph{deterministic} parallelism, and how it differs from concurrency).
Let us see what happens in the probabilistic case.
\newcommand{\xor}{\mathtt{~XOR~}}
\newcommand{\two}{\frac{1}{2}}
\begin{exa}[Confluence failure]\label{ex:motivation}
	Let us consider  the  untyped $\lambda$-calculus extended with a binary operator $\oplus$ which models probabilistic choice.
	Here $\oplus$ is  just flipping a fair coin:
	$M\oplus N$  reduces to either $M$ or $N$ with equal probability $1/2$; we write this
	as $M\oplus N \red \{M^{\frac{1}{2}},N^{\frac{1}{2}}\}$.

	Consider the term $PQ$, where $P=(\lam x. x)(\lam x.x \xor x)$ and $Q=(\true \oplus \false)$; here $\xor$ is the  standard constructs for the exclusive $\mathtt{OR}$,
	$\true$ and $\false$ are terms which encode the booleans. % and $\oplus$ is a probabilistic construct which models a fair coin.
	\begin{itemize}
		\item If we evaluate $P$ and $Q$ independently, from $P$ we obtain $\lam x.(x \xor x)$, while from $Q$ we have either $\true$ or $\false$, with equal probability $1/2$.
		By composing the partial results, we obtain $\{(\true \xor \true)^{\two}, (\false \xor \false)^{\two} \}$, and therefore   $\{\false^1\}$.

		\item If we evaluate $PQ$ sequentially, in a standard leftmost-outermost fashion, $PQ$ reduces to $(\lam x.x \xor x)Q$ which reduces to
		$(\true \oplus \false) \xor (\true\oplus \false)$ and eventually  to $\{\true^\two, \false^{\two}\}$.
		%$\true$ with probability  $1/2$  and $\false$ with probability $1/2$.
	\end{itemize}
\end{exa}
\begin{exa}\label{ex:motivation2}
	The situation  becomes even more complex if we examine also the possibility of diverging; try the same experiment  on the term $PR$, with $P$ as above, and  $R=(\true\oplus \false)\oplus \Delta\Delta$ (where $\Delta=\lam x.xx$). Proceeding as before,
	we now obtain either $\{\false^{\frac{1}{2}}\}$ or
	$\{\true^{\frac{1}{8}},\false^{\frac{1}{8}}\}$.
\end{exa}

We do not need to loose the  features of $\lambda$-calculus in the \emph{probabilistic} setting.
In fact, while some care is needed,
determinism of the evaluation \emph{can be relaxed} without giving up uniqueness of the result: the calculus we introduce in Section~\ref{sec:weak} is an example (we relax determinism to Random Descent); we fully develop this direction in further  work~\cite{FaggianRonchi}.
To be able to do so,  we \emph{need abstract
	tools and proof techniques} to analyze \emph{probabilistic} rewriting.
The same need for theoretical tools holds, more in general, whenever we desire to have a probabilistic  language which
allows for \emph{deterministic parallel reduction}. %Again, we need  general theoretical tools.

In this paper we focus  on \emph{uniqueness of the result}, rather than  confluence. While important, confluence is a sufficient but  {not necessary}  property to have uniqueness of normal forms.

\subsubsection{Other key notions}\label{sec:lambdaintro}\label{background}

 \paragraph{Confluence is  not enough.}
Key to non-deterministic evaluation strategies is that, despite the fact that there are many ways of evaluating
a term,  \emph{all choices eventually yield the same result}.
To this aim, confluence is not enough.
The reduction of a term that has a normal form may still produce  diverging computations, which yield  \emph{no result } (think of $\beta$-reduction in usual $\lam$-calculus, reducing the term $(\lam x.z)(\Delta\Delta)$).
What we really want for a non-deterministic evaluation strategy %$\ered$
is that all  reduction sequences from the same   $\tm$ have \emph{the same behaviour}: if  $\tm$ has a  normal form, then
\emph{all} reduction sequences from $\tm$ eventually reach it (uniform normalization);
ideally, all should do so
in the \emph{same number of steps}. {This latter  property  is known as
Random Descent~\cite{Newman, Oostrom07,OostromT16}, and it is often  guaranteed in the literature of $\lambda$-calculus via a diamond-like property.  We will lift these notions  to the probabilistic and asymptotic setting.}

\paragraph{Random Descent.}
Newman's  Random Descent   (RD)~\cite{Newman} is an ARS property which
guarantees that  normalization suffices to establish
both termination and uniqueness of normal forms. Precisely,
if an ARS has random descent, paths  to a normal form do not need  to be unique, but they have \emph{unique length}.
In its essence:
\emph{if a normal form exists, all rewrite sequences  lead to it, and all  have the \emph{same length}\footnote{Or, in   Newman's original terminology: the end-form is reached by \emph{random descent}
		(whenever $x\red^k y$ and $x\red^n u$  with $u$  in normal form, all maximal reductions from $y$ have length  $n-k$ and end in $u$).}}.
While only  few systems directly verify it,  RD is a powerful ARS tool;
a  {typical use}  in the  literature  is to prove that \emph{a strategy} has RD, to conclude that it is \emph{normalizing}.  A well-known property which implies RD is  a  form of diamond: $\leftarrow\cdot \rightarrow ~\subseteq~ (\rightarrow \cdot \leftarrow)  ~\cup ~=$.

 Von Oostrom~\cite{Oostrom07} has  defined  a  characterization of RD by means of a \emph{local}
property, proposing  RD as a uniform method to (locally) compare strategies for normalization and minimality (resp.\ perpetuality and maximality).   Such a  method has then been extended in~\cite{OostromT16}, where the notion of length is abstracted into a notion of measure.
In  Section~\ref{sec:RD_pars} and~\ref{sec:comparing_pars} we develop  similar methods in a \emph{probabilistic} setting.
The  probabilistic analogous of  \emph{length}, is the  \emph{expected number of steps} (Section~\ref{sec:meantime}).

\paragraph{Weak Call-by-Value   $\lambda$-calculus (and its probabilistic counter-part).}
A notable example of system which satisfies  Random Descent is
   Call-by-Value (CbV) $\lambda$-calculus  endowed with    weak evaluation.

 In Plotkin's Call-by-Value $\lam$-calculus, $\beta$-redexes are  fired  only when the argument is a \emph{value} (\ie, a variable or  a $\lambda$-abstraction).  Since the goal is to compute \emph{values}---as is natural in functional programming---evaluation is often restricted to be  \emph{weak}~\cite{Howard,Hindley},
 where weak means no reduction  in the function bodies (\ie within the scope of $\lambda$-abstractions). Weak CbV  is the basis of the  ML/CAML family of functional  languages---and of most probabilistic functional languages.
 There are  three main weak schemes: reducing from left to right, as originally defined by Plotkin~\cite{PlotkinCbV}, from right to left, as in  Leroy's ZINC abstract machine~\cite{Leroy-ZINC} (resulting in a more efficient implementation),
 or in an \emph{arbitrary order}, used for example  in~\cite{LagoM08}.
 While left and right reduction are  deterministic, %and  incomparable,
 weak reduction  in arbitrary order is \emph{non-deterministic} and \emph{subsumes} both.

 If we consider  programs (closed terms), values are exactly the normal forms of weak reduction.
Because it satisfies Random Descent,
 CbV weak reduction $\wred$ has  \emph{striking properties} (see e.g.~\cite{LagoM08} for an account).
First,  if $M$ reduces to a value ($M\wred^* V$), then  \emph{any}
sequence of $\wred$-steps from $M$ will reach $V$; second, the number $n$ of steps such that
$M\wred^n V$ is always the same.

In  Section~\ref{sec:weak}, we study a probabilistic extension of  weak CbV, $\Lambda_\oplus^\weak$. We  show that it has  analogous properties to its classical counterpart:
all  rewrite  sequences converge to the same result, in  the same \emph{expected} number of steps.

\paragraph{Local vs global conditions.}\label{sec:local}
An important distinction in rewriting theory is between local
and global properties.
A property of a term $t$   is global if it is quantified over all rewrite sequences from $t$, it is
local if it is quantified  only over \emph{one-step reductions} from the term. Local properties are easier to test, because the analysis (usually) involves a finite number of cases.
To work  locally---that is, reducing a test problem which is  global  to local properties---dramatically  reduces the space of search when  testing. Let us exemplify this with a familiar example.
%
%To  work \emph{locally} means to reduce a test problem which is  global, \ie,
%quantified over \emph{all rewrite sequences} from a term, to local properties (quantified only over \emph{one-step reductions} from the term), thus     reducing the space of search when  testing.
%

A  paradigmatic example of global property is confluence (CR): $b  \leftstar\leftarrow  a \rightarrow^* c$  $\implies$ $\exists d$ s.t.
$b\rightarrow^* d \leftstar \leftarrow c$.
%is quantified over \emph{all rewrite sequences} from a term; it is a difficult property to establish.
Its   global nature makes it
difficult  to establish.
A standard way  to factorize the problem is: (1.)  prove termination and (2.) prove  \emph{local} confluence (WCR): $b \leftarrow  a \rightarrow c$ $\implies$ $\exists d$ s.t.
$b\rightarrow^* d \leftstar \leftarrow c$.
This is  exactly
\emph{Newman's lemma:~~Termination + WCR $\implies$ CR\@.}
The beauty  of  Newman's lemma is that  a global property (CR)  is guaranteed by a local property (WCR).

Locality is also the strength and beauty of the Random Descent method. While Newman's lemma fails in a probabilistic setting, Random Descent methods adapt well.

\subsection{Related work}

First, let us observe that there is a vast
literature on probabilistic transition systems,
however   objectives and therefore  questions and  tools  are different than those of PARS\@. A similar distinction exist  between abstract rewrite systems  and  transition systems.
Here we  discuss related work in the context of  PARS~\cite{BournezG05,BournezK02}.

We are not aware of any work which investigates    \emph{normalizing strategies} (or   \emph{normalization} in general,  rather than   termination).
Instead, \emph{confluence} in probabilistic rewriting has already drawn   interesting work.  A notion of {confluence}  for a  probabilistic rewrite system defined over
a $\lambda$-calculus is studied in~\cite{DiazArrighiGG,DalLagoMZ11}; in both cases, the probabilistic behaviour  corresponds to  measurement in a quantum system. %(is an inherently a probabilistic operation.)
The work more closely related to our goals is~\cite{DiazMartinez17}. It studies confluence of non-deterministic PARS   in the case of {finitary termination}
(being finitary is the reason why   Newman's Lemma holds), and in the  case of  \AST\@. As we observe in Section~\ref{sec:unique},
their notion of unique limit distribution
(if $\alpha, \beta$ are limits, then $\alpha =\beta$), while \emph{simple}, it is \emph{not} an analogue of \UN\ for \emph{general} PARS\@.
We extend the analysis beyond \AST, to the general case, which arises naturally when  considering \emph{untyped} probabilistic $\lambda$-calculus.
On confluence, we also mention~\cite{Kirkeby}, whose results however do not cover  \emph{non-deterministic PARS};  the probability of the  limit distribution is  concentrated in  a single element, in the spirit of Las Vegas Algorithms.~\cite{Kirkeby}
revisits results from~\cite{BournezK02}, while we are  in the non-deterministic  framework of~\cite{BournezG05}.

The way we define the \emph{evolution of a PARS}, via the one-step relation $\redd$,  follows the   approach
in~\cite{popl}, which also  contains
an embryo of the current work (a form of diamond property);
the other results and developments are novel.
A technical difference with~\cite{popl} is that for the formalism  to be general, a refinement  is  necessary
(see Section~\ref{sec:evolution});
the issue was  first pointed out in~\cite{DiazMartinez17}.
Our refinement is  a variant
of the one introduced (for the same reasons) in~\cite{Avanzini}; there, normal forms are discarded---because the authors are only interested in the probability of termination---while we are interested in a more  qualitative analysis of the result.~\cite{Avanzini} demonstrates the  equivalence with the approach in~\cite{BournezG05}.

\bigskip
\emph{Quantitative Abstract Rewrite Systems (QARS)} refine Ariola and Blom's notion of \emph{Abstract Rewrite Systems with Information content} (ARSI)~\cite{AriolaBlom02}; there, to the  ARS is associate a \emph{partial order} which expresses a
comparison between the ``information content'' of the elements. Here, we simply move  from partial orders to $\omega$-complete partial orders ($\omega$-cpo).
The difference is in the notion of limit,  hence its properties, and our novel contribution is the study of such properties.
ARSI  are tailored   to  infinite normal forms in the sense of   B{\"o}hm and Levy-Longo trees---limits (infinite normal forms) are there given by completing the partial order
 via a specific standard construction, ideal completion  (see for instance Ch.~1 in~\cite{AmadioCurien}). So, given an element $\tmt$ in an ARSI,  the   infinite normal form of $\tmt$ is the downward closure of
 the set of the information contents of all  its reducts.
Such an approach  would not suit probability distributions, but moving to $\omega$-cpo suffices.  Being simply the supremum of an $\omega$-chain, the  notion of limit which come with  QARS is more general\footnote{{Notice that the \emph{ideal completion} of a partial order is in particular an $\omega$-cpo.}} and flexible, allowing us to model a larger variety of situations.  All results we establish for limits in the setting of QARS also hold  for the infinite normal forms of ARSI, while the converse is not true.
 In Appendix~\ref{app:ARSI} we give a concrete example that shows  the difference:  a confluent ARSI has unique infinite normal forms (Theorem 5.4 there)---the analogue result is (in general) not true for QARS\@.

\section{Probabilistic Abstract Rewriting System}\label{sect:distrpars}\label{sec:pars}
%\section{Basics}
We assume the reader familiar with the basic notions of rewrite theory (such as Ch.~1 of~\cite{Terese03}), and of   \emph{discrete} probability theory. % (see e.g.~\cite{Bremaud}).
We  review  the basic language of both.
We then   recall the     definition of  \emph{probabilistic abstract rewrite system} from~\cite{BournezK02,BournezG05}---here denoted
\pars---and  explain on  examples  how  a system described by a \pars    evolves. This will motivate the
formalism which we present  in Section~\ref{sec:formalism}.

\subsection{Basics on ARS}\label{sec:ARS}
\newcommand{\NFP}{\texttt{NFP}}
  An \emph{abstract rewrite system (ARS)} is a pair
$\C= (C,\red)$ consisting of a  set $C$ and a binary
relation $\red$ on $C$ (called reduction) whose pairs are written  $t \to s$ and called \emph{steps};  $\red^*$ (resp. $\red^=$)  denotes the transitive reflexive (resp.\ reflexive) closure of $\red$.  We  write $c\not \red $ if
  there is no $u$ such that
$c\red u$; in this case, $c$ is a \textbf{normal form}.
 $\Term{\C}$ denotes   the set of the normal forms of $\C$.
 If  $c\red^*u$ and $u\in \Term{\C}$, we say $c$ has a normal form $u$.

A relation $\red$ is \emph{deterministic} if for each $t\in \AA$ there is at most one $s\in \AA$ such that $t\red s$.

\paragraph{Unique Normal Form}
 $\C$ has the property of  \textbf{unique normal form} (\emph{with respect to reduction}) (\UN) if $\forall c\in C, \forall  u, v\in \Term{\C},  \big(c\red^*u ~\&~ c\red^*v \implies  u=v \big)$.
$\C$ has the \textbf{normal form property} (\NFP) if $\forall b, c\in C, \forall u\in \Term{\C},   \big(b\red^*c ~\&~ b\red^*u \implies c\red^*u \big)$.  Clearly,
\NFP\ implies \UN\ (and confluence implies \NFP).

\paragraph{Normalization and Termination}
The fact that an ARS has unique normal forms does not imply neither that all elements have a normal form, nor that if an element has a normal form, each  rewrite  sequence converges to it.
An element $c$ is  \textbf{terminating}\footnote{Please observe  that the   \emph{terminology is community-dependent}.
	In  logic:  Strong Normalization, Weak Normalization, Church-Rosser (hence the \emph{standard abbreviations} \SN, \WN, CR).
	In computer science:  Termination, Normalization, Confluence.} (aka  \textbf{strongly normalizing}, \SN),  if it has no infinite sequence $c\red c_1 \red c_2 \ldots$; it is   \textbf{ normalizing} (aka \textbf{weakly normalizing}, \WN), if  it has a normal form.
{These are all important  properties to establish  about an ARS, as it is important to have a rewrite strategy which  finds a normal form, if it exists.}

	\subsection{Basics on Probabilities}\label{sec:proba}
%	\paragraph{Basics on Probabilities.}\label{sec:proba}
{
The intuition is that random phenomena are observed by means of experiments (running a probabilistic program is such an experiment); each experiment results in an outcome. The collection of all possible outcomes  is represented  by a set, called the \textbf{sample space} $\Omega$.
}
When the sample space $\Omega$ is \emph{countable}, the theory is   simple.
A \emph{discrete probability space} is given by a pair $(\Omega, \mmu)$,
where  $\Omega$ % called the \emph{sample space},
 is  a \emph{countable} set, and $\mmu$ is  a \textbf{discrete probability distribution}
on  $\Omega$, \ie\  a
function $\mmu:\Omega \to [0,1]$ such that  $ \sum_{\omega\in \Omega} \mmu(\omega) = 1$.
A probability measure is assigned  to any  subset  $A\subseteq \Omega$ as $\mmu(A)=\sum_{\omega\in A} \mmu(\omega)$.
In the language of probabilists, a {subset} of $\Omega$ is called an \emph{event}.
\begin{exa}[Die]\label{ex:odd}
	Consider tossing a die once. The  space of possible outcomes is the
		set $\Omega = \{1, 2, 3, 4, 5, 6\}$. The probability  $\mmu$ of each outcome is $1/6$. The event \emph{``result is odd''} is the subset $A= \{1, 3, 5\}$, whose probability  is $\mmu(A)= 1/2$.
\end{exa}

Each \emph{function} $F : \Omega \to \Delta$, where $\Delta$ is another countable set,
\textbf{induces a  probability distribution} $\mmu^F$ on $\Delta$
by composition:  $\mmu^F(d) := \mmu (F^{-1}(d))$ \ie\ $\mmu(\{\omega \in \Omega: F(\omega) = d\})$. Thus $(\Delta, \mmu^F)$ is also  a
probability space.
In the language of probability theory, $F$ is called a \emph{discrete random variable} on $(\Omega, \mmu)$.
The  \textbf{expected value} (also called the expectation or  mean)
of a random variable $F$ is the weighted (in proportion to probability) average of the possible values of $F$.
{Assume  $F:\Omega \to \Delta$  discrete and $g: \Delta\to \Real$ a non-negative function, then   $E(g(F))=\sum_{d\in \Delta}g(d) \mmu_F(d)$.}

\subsection{(Sub)distributions:  operations and notation}\label{sec:dist} % chktex 36
We need  the notion of subdistribution to account for   partial results, and for unsuccessful computation.
Given a countable set $\Omega$,  and a  function $\mmu:\Omega\to[0,1]$, we define   $ \norm \mu :=\sum_{\omega\in \Omega} \mmu(\omega)$. The function $\mmu$
is a probability \textbf{subdistribution} if
$\norm \mmu  \leq 1$.  We write  $\DST{\Omega}$ for the set
of   subdistributions on $\Omega$.
%With a slight abuse of language, we sometime  use the term distribution  also for   subdistribution.
The  \emph{support} of  $\mmu$ is the set
$\supp{\mmu}=\{a\in \Omega \mid\mmu(a)>0\}$.
$\DSTF{\Omega}$
denotes the set of $\mmu\in \DST{\Omega}$ with \emph{finite support}, and $\zero$ indicates the subdistribution of empty support.

$\DST{\Omega}$ is equipped with the pointwise  \textbf{order relation} of functions:  $\mmu \leq \rho$ if
$\mmu (a) \leq \rho (a)$ for each $a\in \Omega$.
\textbf{Multiplication} for a scalar ($p\cdot \mmu$) and \textbf{sum} ($\sigma + \rho$) are defined as usual,
$(p\cdot \mmu) (a)= p\cdot \mmu(a)$,   $(\sigma+\rho)(a)=\sigma(a)+\rho(a)$, provided $p\in[0,1]$, and $\norm \sigma+\norm \rho \leq 1$.

\begin{nota}[Representation]\label{notation:dist} We represent a (sub)distribution by explicitly indicating the support, and (as superscript) the probability assigned to each element by $\mmu$. We  write $\mmu=\{a_0^{p_0}, \dots,  a_n^{p_n}\}$ if $\mmu(a_0)=p_0,\dots, \mmu(a_n)=p_n$ and  $\mmu(a_j)=0 $ otherwise. % chktex 36
\end{nota}

%%%%%%%%%%%%%%%%%%%%%%%%%%%%%%%%%%%%%%%%%%%%%%%%%%%%%%%%%%%

\subsection{Probabilistic Abstract Rewrite Systems (\PARS)}

\begin{figure}
	\centering
		\fbox{
		\begin{minipage}[b]{0.28\textwidth}
			{\scriptsize
				$r_0:c\red\{c^{1/2}, \true^{1/2}\}$

				\begin{forest}
					L/.style={
						edge label={node[left,blue,font=\tiny]{#1}}
					},
					for tree={
						grow=0,reversed, % tree direction
						parent anchor=east,child anchor=west, % edge anchors
						edge={line cap=round},outer sep=+1pt, % edge/node connection
						l sep=8mm % level distance
					}
					[c,
					[c,  L={1/2}, 	[c,  L={1/4},[\dots],[\true]]
					[\true,  L={1/4}]]
					[\true,  L={1/2}]
					]
				\end{forest}	}
			  \captionsetup{width=\linewidth}
			\caption{Almost Sure Termination}\label{fig:AST}
	\end{minipage}
}
	\fbox{
		\begin{minipage}[b]{0.22\textwidth}
			{\tiny
				\begin{forest}
					P/.style={
						edge label={node[left,blue,font=\tiny]{#1}}
					},
					for tree={
						grow=0,reversed, % tree direction
						parent anchor=east,child anchor=west, % edge anchors
						edge={line cap=round},outer sep=+1pt, % edge/node connection
						l sep=8mm % level distance
					}
					[2,
					[1,P={1/2}[0,P={1/4}],[2~~$\cdots$,P={1/4}]],
					[3,P={1/2}[2~~$\cdots$,P={1/4}],[4~~$\cdots$,P={1/4}]]
					]
				\end{forest}
			}
		                \captionsetup{width=\linewidth}%
			\caption{Deterministic \PARS}%
			\label{fig:walk}
		\end{minipage}
	}
	\fbox{
		\begin{minipage}[b]{0.31\textwidth}
			{\scriptsize
				\begin{forest}
					P/.style={
						edge label={node[left,blue,font=\tiny]{#1}}
					},
					for tree={
						grow=0,reversed, % tree direction
						parent anchor=east,child anchor=west, % edge anchors
						edge={line cap=round},outer sep=+1pt, % edge/node connection
						l sep=8mm % level distance
					}
					[2,
					[1,P={1/2}[0,P={1/4}],[2,P={1/4},[1~~$\cdots$,P={1/8}],[3~~$\cdots$,P={1/8}]]],
					[3,P={1/2}[2,P={1/4},[stop,P={1/4}]],[4~~$\cdots$,P={1/4}]]
					]
				\end{forest}
			}
		    \captionsetup{width=\linewidth}
		\caption{Non-deterministic \PARS}%
			\label{fig:walk_stop}
		\end{minipage}
	}
\end{figure}
%!!!!
A \emph{probabilistic abstract rewrite system (\PARS)} is a pair
$(A,\red)$  of a countable set $A$ and a
relation ${\red}\subseteq{A\times\DSTF{A}}$  such that for each $(a,\beta)\in{\red}$, $\norm{\beta}=1$. We  write  $a \red \beta$ for  $(a,\beta)\in{\red}$ and we call it  a \emph{rewrite step}, or a \emph{reduction}. An element $a\in A$ is
in \emph{normal form}  if there is no $\beta$ with
$a\red \beta$. %, which we write $a \not \red$.
 We denote by $\Term{\A}$   the set of the normal forms  of $\A$  (or simply $\Term{}$ when $\A$ is clear).
A \PARS is  \textit{deterministic}  if, for all $a$, there is at most one $\beta$ with $a\red \beta$.

{\begin{rem}
The intuition behind  $a\red \beta$ is that the rewrite step $a\red b$ ($b\in A$)
has probability  $\beta(b)$. The total probability given by the sum of all steps $a\red b$ is  $1$.
\end{rem}}

\paragraph{Probabilistic vs Non-deterministic.} It is important to understand the distinction between  probabilistic choice (which \emph{globally  happens with certitude})
and non-deterministic choice (which leads to   different distributions of outcomes.)
Let us discuss some examples.
\begin{exa}[A deterministic \PARS]\label{ex:walk}   Fig.~\ref{fig:walk} shows a simple random walk over $\Nat$, which describes
	 a gambler starting with $2$ points and playing a game where every time he either  gains $1$ point  with probablity $1/2$ or  looses $1$ point   with probability $1/2$. This system is   encoded  by  the following \PARS on $\Nat$:
	$n+1 \red \{n^{1/2}, (n+2)^{1/2}\}$. Such a
	\PARS is  \emph{deterministic}, because for every element, at most one choice applies.  Note that $0$ is the (only) normal form.
\end{exa}
\begin{exa}[A non-deterministic \PARS]\label{ex:walk_stop} Assume now (Fig.~\ref{fig:walk_stop})  that the gambler of  Example~\ref{ex:walk} is also given the possibility to stop at any time. The two choices are here encoded as follows:
	\begin{equation*}
			n+1 \red \{n^{1/2}, (n+2)^{1/2}\},\quad  n+1 \red \{\texttt{stop}^1\}
	\end{equation*}
The choice between two possible rules makes the system non-deterministic, and therefore the system can evolve in several different ways.
Fig.~\ref{fig:walk_stop} illustrates one possible way.
\end{exa}

{
\subsection{Evolution of a system described by a \PARS\@.}\label{sec:evolution}
}
 We now  need to explain how  a system which is described by a \PARS evolves.
 An option  is to follow the stochastic evolution of a single run, \emph{a sampling  at a time}, as we have done in Fig.~\ref{fig:AST},~\ref{fig:walk}, and~\ref{fig:walk_stop}.
 This is the approach in~\cite{BournezG05}, where non-determinism is  solved by the use of policies.
  {Here we follow a different (though equivalent) way.}
 We  describe the possible states of the  system, at a certain time $t$, \emph{globally}, essentially as a  distribution on the space of all elements.
 The evolution of the system is then a  sequence of  such states. Since all the probabilistic choices are taken together, a global step happens with probability $1$; the only source of non-determinism in the evolution of the system is choice.
    This global approach allows us to deal with non-determinism by using techniques which have been developed in Rewrite Theory.
 Before introducing the formal definitions, we informally examine some examples, and point  out why some care is needed.

 \begin{figure}\centering
 	%	\fbox{
 	%	}
 	\fbox{
 		\noindent
 		\begin{minipage}[t]{0.45\textwidth}
 			{\tiny
 				$r_0: a\red\{a^{1/2}, \true^{1/2}\},~~  r_1: a \red\{a^{1/2}, \false^{1/2}\} $

 				\begin{forest}
 					L/.style={
 						edge label={node[midway,left, font=\tiny]{#1}}
 					},
 					%for tree = {sn edges, grow'=0, l=2.5cm, s sep=0.2cm, anchor=west, child anchor=west}
 					for tree={
 						grow=0,reversed, % tree direction
 						parent anchor=east,child anchor=west, % edge anchors
 						edge={->},outer sep=+1pt, % edge/node connection
 						%	rounded corners,minimum width=15mm,minimum height=8mm, % node shape
 						l sep=6mm % level distance
 					}
 					[$\{a^1\}$,
 					[{$\pmb{\{a^{1/2},\true^{1/2}\}}$},  blue, L={$r_0$}	[{$\pmb{\{a^{1/4},\true^{3/4}\}}$} $\cdots$, blue, L={$r_0$}]	[{$\{a^{1/4},\true^{1/2},\false^{1/4}\} $, }$\cdots$, L={$r_1$}, ]]
 					[{$\{a^{1/2},\false^{1/2}\}$}, L={$r_1$},   	[{$\{a^{1/4},\true^{1/4},\false^{1/2}\} $} $\cdots$, L={$r_0$} ]	[{$\{a^{1/4},\false^{3/4}\}$}$\cdots$, L={$r_1$}, ]]
 					]
 				\end{forest}
 			}\captionsetup{width=\linewidth}  \caption{Ex.\ref{ex:continuum} (non-deterministic \PARS)}\label{fig:continuum}
 		\end{minipage}
 	}
 	%	\hspace{1cm}
 	\fbox{
 		\noindent
 		\begin{minipage}[t]{0.45\textwidth}
 			{\tiny
 				$r_0: a\red\{a^{1/2}, \true^{1/2}\},~~ r_2: a\red \{a^1\}$

 				\begin{forest}
 					L/.style={
 						edge label={node[midway,left, font=\tiny]{#1}}
 					},
 					%for tree = {sn edges, grow'=0, l=2.5cm, s sep=0.2cm, anchor=west, child anchor=west}
 					for tree={
 						grow=0,reversed, % tree direction
 						parent anchor=east,child anchor=west, % edge anchors
 						edge={->},outer sep=+1pt, % edge/node connection
 						%	rounded corners,minimum width=15mm,minimum height=8mm, % node shape
 						l sep=6mm % level distance
 					}
 					[$\{a^1\}$,
 					[{$\pmb{\{a^{1/2},\true^{1/2}\}}$}, blue,  L={$r_0$} 	[{$\pmb{\{a^{1/4},\true^{3/4}\}}$} $\cdots$,blue, L={$r_0$}]	[{$\{a^{1/2},\true^{1/2}\} $}$\cdots$, L={$r_2$}]]
 					[{$\{a^1\}$}, red, L={$r_2$}	[{$\{a^{1/2},\true^{1/2}\}$} $\cdots$, L={$r_0$} ]	[{$\{a^1\}$} $\cdots$, red, L={$r_2$}]]
 					]
 				\end{forest}
 			}\captionsetup{width=\linewidth}  \caption{Ex.\ref{ex:dyadic} (non-deterministic \PARS)}\label{fig:dyadic}
 		\end{minipage}
 	}
 \end{figure}

\begin{exa}[Fig.\ref{fig:AST} continued]\label{ex:AST}
	The \PARS described by the  rule $r_0:c\red\{c^{1/2},\true^{1/2}\}$  (in Fig.~\ref{fig:AST})  evolves as follows: $\{c\},\{c^{1/2},\true^{1/2}\},
	\{c^{1/4},\true^{3/4}\},\dots$.
\end{exa}
\begin{exa}[Fig.\ref{fig:continuum}]\label{ex:continuum}
Fig.~\ref{fig:continuum} illustrates the possible evolutions of a non-deterministic system  which  has two rules: $r_0:a\red\{a^{1/2},\true^{1/2}\}$ and $r_1:a\red\{a^{1/2},\false^{1/2}\}$. The arrows are annotated with the chosen rule.
\end{exa}

\begin{exa}[Fig.\ref{fig:dyadic}]\label{ex:dyadic}
	Fig.~\ref{fig:dyadic} illustrates the possible evolutions of a system with rules $r_0:a\red\{a^{1/2},\true^{1/2}\}$ and
	$r_2:a\red\{a^1\}$.
\end{exa}
%This approach demands some care.
If we
look at Fig.~\ref{fig:walk_stop},  we observe  that after two steps, there are \emph{two distinct occurrences} of the element \texttt{2}, which live in \emph{two different runs} of the program: the run  \texttt{2.1.2}, and the run  \texttt{2.3.2}.
There are two possible transitions from each  $2$. The next transition only depends on the fact of having   \texttt{2},
not on the  run in which \texttt{2} occurs: its history is only a way to distinguish the  occurrence.
%
%The history ot the run is only a way to distinguish different occurrences of the same state.
For this reason, given a \PARS $(A,\red)$,
we keep track of \emph{different occurrences} of an element $a\in A$, but not necessarily of the history.
Next section formalizes these ideas.

{
\paragraph{Markov Decision Processes.} To understand our distinction  between occurrences of  $a\in A$ in different paths,  it is helpful to think how  a system  is described in the framework of Markov Decision Processes (MDP)~\cite{Puterman94}. Indeed,
 in the same  way as  ARS  correspond to transition
	systems, \PARS correspond to probabilistic transitions.
Let us regard a \PARS  step $r: a\red \beta$ as a probabilistic transition ($r$ is here a name for the rule). Let assume $a_0\in\A$ is an initial state.
In the setting of MDP, a typical element (called \emph{sample path}) of the sample space $\Omega$ is a sequence $\omega=(a_0,r_0,a_1,r_1 \ldots)$ where   $r_0:a_0\red \beta_1$ is a rule, $a_1\in \supp{\beta_1}$ an element, $r_1:a_1\red \beta_1$, and so on.
 The index  $t=0,1,2,\dots,n,\dots$  is interpreted as \emph{time}. On $\Omega$ various  random variables are defined; for example,  $X_t=a_t$, which represents the state at time $t$.  The sequence $\langle X_t\rangle$ is called  a  stochastic process.
 %The the probability distribution on the system state is similarly  defined. %$P_1(-)$ denotes the initial distribution of the system state.
}

\section{A Formalism for Probabilistic Rewriting}\label{sec:formalism}

%To syntactically represent  the global evolution of a probabilistic terms, we use  the notion of multidistribution~\cite{Avanzini,FaggianRonchi}.

This section presents  a formalism   to  describe the global evolution of a  system described by a \PARS, which is a variant of  that used in~\cite{Avanzini}.  The  equivalence with the approach   in~\cite{BournezG05} is demonstrated in~\cite{Avanzini}.

\subsection{PARS}
Let  $A$   be a countable set on which a \PARS $\A=(A,\red)$ is given. We define a rewrite system $(\MA,\redd)$, where $\MA$ is the set of objects to be rewritten, and $\redd$ a relation on $\MA$. We indicate as PARS the resulting rewriting system.

\paragraph{The objects to be rewritten.} $\MA$ is the  set of all \emph{multidistributions} on $A$, which are defined as follows.
Let    $\m$ be  a multiset\footnote{A \emph{multiset} is  a (finite) list of elements,  modulo reordering.}  of pairs of the form $p a$, where $p\in \lopen{0,1}$ is a real number, and $a\in A$ an element of $A$; the multiset
 $\m=\mset{p_i a_i\mid i\in I}$ is a multidistribution  on $A$ if
$\norm \m = \sum_{\iI} p_i \leq 1$.
We  write the multidistribution $\mdist{1a}$ simply as $\mdist{a}$.

\emph{Sum and product} are partial operations, similarly to what happens for distributions.
The  {sum}  of  multidistributions  is  denoted by $\dsum$, and it is the disjoint union of multisets (think of list concatenation). Given two  multidistributions $\m_1=\mset{p_i a_i\mid i\in I}$ and $m_2= \mset{q_{j}b_{j}\mid j\in J}$, their sum $ \mset{p_i a_i\mid i\in I} \uplus \mset{q_{j}b_{j}\mid j\in J}$  is defined only if %it is a multidistribution  (\ie $\norm {\m_1} + \norm {\m_2} \leq 1$).
$\norm {\m_1} + \norm {\m_2} \leq 1$.
The product $q\cdot \m$ of a scalar $q$ and a multidistribution $\m$ is defined pointwise, provided that $p\in[0,1]$: $q\cdot \mset{ p_{1}a_{1}, \ldots , p_{n}a_{n}}=\mset{ (qp_{1})a_{1}, \ldots , (qp_{n})a_{n}} $.

Intuitively, a multidistribution $\m\in \MDST{A} $ is a \emph{syntactical representation}  of a discrete probability space
where  each point   in the space (each outcome)  is associated to a probability and an element of $A$. More precisely, \emph{each pair in $\mm$ correspond to
a trace of computation}, or---in the language of Markov Decision Processes---to \emph{a sample path}.

\paragraph{The rewriting relation.}%\label{sec:reductions}\label{sec:sequences}
\newcommand{\OneStep}{\texttt{OneStep}}

The  binary relation $\redd$ on
$\MA$ is obtained by lifting the relation $\red$ of the \PARS $\A=(A,\red)$, as follows.
\begin{defi}[Lifting]\label{def:lift}
	Given a relation $\to\subseteq A\times \DST{A}$, its lifting to a relation $\redd\subseteq\MA\times\MA$ is defined  by the   rules
			\[
			\infer[L1]{\mset{a}\redd \mset{a}}{a\not \to} \quad \quad
			\infer[L 2]{\mset{a}\Red \mset{p_k a_k \mid \kK}}{a\to\{a_k^{p_k} \mid \kK\}}  \quad  \quad
			\infer[L3]{ \mset{p_{i}a_{i}\mid i\in I} \redd  \sum_{\iI} {p_i\cdot \m_i}}
			{\big(\mset{a_i} \redd  \m_i\big)_{\iI} }
			\]

\end{defi}
 For the lifting, several natural choices are possible.    Here we  force  \emph{all} non-terminal elements to be reduced. This choice plays an important role for the development of the paper, as it corresponds to  the key notion of \emph{one step} reduction in classical ARS (see  discussion in Section~\ref{sec:discussion}). Let us discuss some more the lifting rules.
\begin{itemize}
	\item Rule $ L1 $.  Note that the relation   $\redd$
 is reflexive on normal forms.
\item  Rule $L2$. Please observe that  $ \mset{p_k a_k \mid \kK} \in \MA$ is simply a  representation of the distribution $\{a_k^{p_k} \mid \kK\}\in \DST{A}$.
\item Rule $L3$. To apply  rule $L3$, we  have to choose a reduction step from  $a_i$ for \emph{each}  $\iI$.
%We could write the  premiss also as  $\big((j,\m_j)^1\Red \mm_j\big)_{j\in J}$.
The (disjoint) sum of all  $\mm_i$ ($i\in I$) is  weighted with  the scalar $p_i$ associated to  each  $p_{i}a_i$.
 \end{itemize}

\begin{exa}
	Let us  derive the reduction in Fig.~\ref{fig:walk_stop}. For readability,  elements in $\Nat$ are in bold.
			\begin{center}
	{\small
		$\infer{[\textbf{2}]\redd [\two\textbf{1}, \two \textbf{3}]}
		{\textbf{2}\red\{\textbf{1}^{\two}, \textbf{3}^{\two}\}}$\quad
		$\infer{[\two\textbf{1}, \two \textbf{3}]\redd[\frac{1}{4}\textbf{0}, \frac{1}{4}\textbf{2},\frac{1}{4}\textbf{2}, \frac{1}{4}\textbf{4}]}
		{  1\red\{\textbf{0}^{\two},\textbf{2}^{\two}\}&3\red\{\textbf{2}^{\two},\textbf{4}^{\two}\} }$\quad
		$\infer{[\frac{1}{4}\textbf{0}, \frac{1}{4} \textbf{2}, \frac{1}{4} \textbf{2}, \frac{1}{4}\textbf{4}] \redd  [\ldots, \frac{1}{4} \textbf{stop},  \frac{1}{8}\textbf{1}, \frac{1}{8}\textbf{3},  \ldots]  }{\dots & \textbf{2}\red \{stop^1\} &
			\textbf{2}\red\{\two\textbf{1},\two \textbf{3}\} &\dots}$}
		\end{center}
\end{exa}

\paragraph{PARS\@.}  We indicate as PARS the rewrite system $ (mA,\Red)$ which is  induced by the \pars $(A,\red)$.

\paragraph{Rewrite sequences.}
We write $\mm_0\redd^*\mm_n$ to indicate that there is a \emph{finite sequence}  $\mm_0, \dots,\mm_n$  such that $\mm_{i} \redd \mm_{i+1}$ for all   $0 \leq i < n$ (and $\mm_0\redd^k \mm_k$ to specify its length  $k$).
%The letter $\S,\T$ range over finite sequences.
We write $\seq \mm$ to indicate an \emph{infinite rewrite  sequence}.

\paragraph{Figures conventions:} we depict \emph{any} rewrite relation simply as $\rightarrow$;  as it is standard, we use $\twoheadrightarrow$ for  $\rightarrow^*$; solid arrows are universally quantified, dashed arrows are existentially quantified.

\subsection{Normal forms and observations}\label{sec:observations}
Intuitively, a multidistribution $\m\in \MDST{A} $ is a \emph{syntactical representation}  of a discrete probability space
where at each element  of the space  is associated a probability and an element of $A$.
This space may contain various  information. We analyze this space by defining random variables that \emph{observe specific properties of interest}.
Here we focus on a specific event of interest:  the set  $\NFA$
of \emph{normal forms} of $\A$.

\bigskip
\paragraph{Distribution over the elements of $A$.}
 First of all, to each multidistribution $\m = \mset{p_i a_i \mid \iI}$ we can associate a (sub)distribution $\muflat \in \DST{A}$ as follows: % chktex 36
\[
\muflat(c) = \sum_{i\in I}  q_i
\qquad\qquad
q_i=\left\{
\begin{array}{ll}
	p_i & \mbox{if $a_i=c$}\\
	0   & \mbox{otherwise}
\end{array}
\right.
\]
Informally, for each $c\in A$, we sum  the probability of all occurrence of $c$ in the multidistribution (observe that, $\m$ being a multiset, there are in general \emph{more
	than one} elements $p_{i}a_i$ where $a_i=c$).

\paragraph{Distribution over the \emph{normal forms} of $A$.}
Given   $\mm\in\MA$, the \textbf{probability that the system is in normal form} %, \ie\ the probability of the subset $\Term{A}$,
is   described by $\muflat(\NFA)$
(recall Example~\ref{ex:odd});
%$\mmu\{(j,t)\mid t\in \Term{\A}\} =\mmu^{\flat}\{\Term{\A}\}$.
the probability that the system is in  a specific normal form $u$ is
described by $\muflat(u)$.

It is convenient to spell-out  a direct definition of both, to which we will refer in the rest of the paper.

\begin{itemize}
	\item The    function
$\nf {-}: \MA \to  \DST{\NFA} \quad
\m \mapsto \nf \mm$
is the restriction of $\muflat$ to $\NFA$.

Informally, this function   extracts  from $\m=\mset{p_i a_i}_{\iI}$ the \emph{subdistribution $\nf \mm$ over normal forms}.

\item The norm  $\norm -: \DST{\NFA}\to [0,1] $  \Big(recall that $\norm \mu =\sum\limits_{u\in \NFA} \mu (u)  $\Big) induces the function
\[\nnorm - : \MA \to  [0,1] \quad
\m \mapsto \nnorm  \mm \]
which  observes the  probability that   $\m$ has reached  a normal form. Clearly, $ \norm{\nf{\mm}} =   \muflat (\Term{\A})$.

\end{itemize}

\begin{exa}
	Let  $\m=\mset{\frac{1}{4} \true, \frac{1}{8} \true, \frac{1}{4} \false, \frac{3}{8} c  }$ (where $\true, \false$ are normal forms,
	and $c$ is not).
	Then
	$\nf \mm=\{\true^{\frac{3}{8}},  \false^\frac{1}{4}  \}$, and  $ \nnorm{\m}=\frac{5}{8}$.
\end{exa}

%%%%%%%%%%%%%%%%%%%%%%%%%%%%%%%%%%%%%%%%%%%%%%%%%

%

The probability of reaching a normal form $u$  can only increase in a rewrite sequence (because of (L1) in Def.~\ref{def:lift}).
Therefore the following key lemma holds.
\begin{lem}\label{lem:basic}
	If $\mm_1\redd \mm_2$ then $\nf\mm_1 \leq \nf\mm_2$ and $\nnorm {\mm_1} \leq \nnorm {\mm_2}$.
\end{lem}

\paragraph{Equivalences and Order.}
In this paper $\mm\in \MA$ is a multiset, for simplicity and  uniformity with~\cite{FaggianRonchi}, but  we could have used lists rather than multisets---as we do in~\cite{Faggian19}. We do not really care of equality of elements in $\MA$---what
 we are  interested are instead  equivalence and order relations w.r.t \emph{the observation of  specific  events}.
For example, the following (recall from Section~\ref{sec:dist} that the order on $\DST{A}$  is the pointwise order):

	Let $\mm,\mr\in \MA$.
	\begin{enumerate}
		\item  \emph{Flat Equivalence}:  $ \eqflat \mm \mr$,
		if $\muflat=  {\mathtt{r}}^{\mathsf{dst}} $.  Similarly,  $\mm \geq_{\flat} \mr $
		if $\muflat \geq  {\mathtt{r}}^{\mathsf{dst}} $.
		\item \emph{Equivalence in Normal Form}:    $\eqnf \mm \mr$,
		if $\nf\mm = \nf\tmr$.  Similarly, $  \relnf \mm \mr  $,
		if $\nf\mm \geq \nf\tmr$

		\item 	  \emph{Equivalence in the $\N$-norm}: $ \eqnorm \mm \mr$,
		if $\norm{\nf\mm} = \norm {\nf\tmr}$, and   $\relnorm \mm \mr$,
		if $\norm{\nf\mm }\geq \norm{\nf\tmr}$
	\end{enumerate}

\noindent
Note  that (2.) and (3.) compare $\mm$ and $\tmr$ abstracting from any element which is not in normal form.
\begin{exa} Assume  $\true$ is a normal form and $a\not=c$ are not.
	\begin{enumerate}
		\item Let $\mm= [\two\true, \two\true],~ \tmr= [1\true]  $.  $\eqflat \mm \tmr$, $\eqnf\mm \tmr$, $\eqnorm\mm \tmr$ all  hold.

		\item Let $\mm=[\two a,  \two \true] $, $\tmr=[\two c,	 \frac{1}{6}\true, \frac{2}{6}\true] $.
		$\eqnf\mm \tmr$, $\eqnorm\mm \tmr$ both \emph{hold},  $\eqflat \mm \tmr$ does \emph{not}.
	\end{enumerate}
\end{exa}
The  above example illustrates  also   the following.
\begin{fact}
	$( \eqflat \mm \mr) ~\implies~ (\eqnf \mm \mr) ~\implies~(\eqnorm \mm\mrho) $.  Similarly for the order relations.
	%  $ \geq_{\flat}(\mm,\tmr) ~\implies~ \relnf \mm \tmr ~\implies~\relnorm\mm\tmr $.
\end{fact}

%%%%%%%%%%%%%%%%%%%%%%%%%%%%%%%%%%%%%%%
\section{Asymptotic Behaviour of PARS}\label{sec:asymptotic}
We  examine   the  asymptotic behaviour of rewrite sequences \emph{with respect to normal forms}, which  are  the most common   notion  of result.

The intuition is that  a rewrite sequence describes  a computation; an element $\mm_i$ such that $\mm\redd^i \mm_i$ represents a state  (precisely, the state at time $i$) in the evolution of the system with  initial state  $\mm$.
%
%%%%%%%%%%%%%%%%%%%%%%%%%%%%%%%%%%%%%%%
The
  \emph{result} of the computation is  a distribution over  the possible normal forms of the probabilistic  program. We are interested in the result when the number of steps tends to infinity, that is \emph{at the limit}.
   This is formalized by  the (rather standard)  notion of \emph{limit distribution} (Def.~\ref{def:limit}). What is new here, is that since
     each element $\mm$  has different possible rewrite sequences (each sequence  leading to a possibly different limit distribution)  to $\mm$ is  naturally associated a \emph{set} of limit distributions.

 A fundamental property for a system such as the $\lam$-calculus is that if an element    has a normal form, it is unique. This is crucial to the computational interpretation of the calculus, because it means that the result of the computation is \emph{well defined}.
A question we need to address in the setting of PARS,  is what does it mean to have  a well-defined result.
With this in mind, we investigate   an analogue of the ARS  notions of normalization, termination, and   unique normal form.

 \subsection{Limit Distributions}\label{sec:limit}
Before introducing limit distributions,  we  revisit  some   facts on  sequences of bounded  functions.

\paragraph{Monotone Convergence.}
We  recall the following standard result.
\begin{thm}[Monotone Convergence for Sums]\label{thm:MCS}
	Let $ X$ be a countable set,  $f_n: X\to[0, \infty]$ a non-decreasing sequence of  functions, such that
	$f(x):= \lim_{n\to\infty} f_n(x)=\sup_n  f_n(x) $ exists for each $x\in X$. Then
	\[\lim_{n\to \infty} \sum_{x\in X} f_n (x) ~=~  \sum_{x\in  X} f(x)\]
\end{thm}

Recall that subdistributions over a countable set $ X$ are  equipped with the \emph{pointwise  order}:  $\alpha \leq \alpha'$ if
$\alpha (x) \leq \alpha' (x)$ for each $x\in X$.
Let $\seq \alpha$ be a \emph{non-decreasing sequence} of (sub)distributions  over $X$. % chktex 36
For each $t\in X$, the sequence $\langle {\alpha_{n}}(t)\rangle_{n\in \Nat}$  of real numbers  is \emph{nondecreasing and bounded}, therefore the sequence  has a limit, which is the supremum: $\lim_{n\to \infty} {\alpha_n}(t)=\sup_n\{{\alpha_n}(t)\}$.
Observe that if $\alpha < \alpha'$ then $\norm \alpha < \norm{\alpha'}$, where we  recall that $\norm \alpha :=\sum_{x\in X} \alpha(x)$.

\begin{lem}\label{lem:MCT} Given $\seq \alpha$ as above, the following properties hold. Define
	\begin{center}
	 	$ \beta(t)= \lim_{n\to \infty}{\alpha_{n}}(t), ~ \forall t\in X$
	\end{center}
\begin{enumerate}

%	\item \begin{equation}\label{eq:norm}\lim_{n\to \infty} \norm{\alpha_n } ~=~   \norm \beta\end{equation}

	\item $ \lim_{n\to \infty} \norm{\alpha_n } ~=~   \norm \beta $

	\item  $\lim_{n\to \infty} \norm{\alpha_n}=\sup_n\{\norm{\alpha_n}\}\leq1$

		\item 	$\beta $  is a \emph{subdistribution} over $X$.

\end{enumerate}
\end{lem}

\begin{proof}  (1.) follows from the fact  that $\seq { \alpha}$  is a nondecreasing sequence of  functions,  hence  (by Monotone Convergence, Thm.~\ref{thm:MCS}) we have:
\begin{center}
		$\lim_{n\to \infty} \sum_{t\in X}{\alpha_n (t) } ~=~   \sum_{t\in X} \lim_{n\to \infty}\alpha_n(t)$
\end{center}
	(2.) is immediate, because the sequence $\langle{\norm{\alpha_n}}\rangle_{n\in \Nat}$ is  nondecreasing and bounded.

    \noindent
	(3.) follows  from (1.) and (2.).
Since
   $\norm \beta = \sup_n \norm{\alpha_n }\leq 1$,  then    $\beta$ is a subdistribution.
\end{proof}

\paragraph{Limit distributions.} Let  $ \A=(mA,\Red)$ be  the rewrite system induced by a \pars $(A,\red)$.

Let $\seq{\mm}$ be a rewrite sequence.
If $t\in \Term \A$, then  $\langle\nf{\mm_{n}}(t)\rangle_{n\in \Nat}$ is nondecreasing (by Lemma~\ref{lem:basic});
so we can apply Lemma~\ref{lem:MCT}, with   $\seq \alpha$ now  being  $\seq {\nf \mm}$.

\begin{defi}[Limits]\label{def:limit} Let $\seq{\mm}$ be   a rewrite sequence from     $\mm\in\MA$. We say
	\begin{enumerate}
			\item  $\seq{\mm}$ \textbf{converges with probability}
		$p = \sup_{n} \{\nnorm{\mm_n}\}$. % \big(written $\seq{\mm} \tolim_{\weight} p$  \big)

		\item   	$\seq{\mm}$ \textbf{converges to}   ${\beta}\in \DST{\Term{\A}}$ %\big  (written $\seq{\mm} \tolim \beta$ \big), where for $t\in \Term{\A}$
\begin{center}
		$\beta(t)= \sup_{n} \{\nf{\mm_{n}}(t)  \} $
\end{center}

	\end{enumerate}
	 We call $\beta$ a \textbf{limit distribution} (on normal forms)  of $\mm$, and $p$ a \textbf{limit probability} (to reach a   normal form) of $\mm$. We write  $\mm\tolim \beta $ (resp. $\mm\tolimp p $) if  $\mm$ has a  sequence which converges to $\beta$ (resp.\ converges with probability $p$). We    define $\Lim(\mm):=\{\beta ~\mid~ \mm\tolim \beta\}$ the set of limit distributions, and  $\pLim(\mm):=\{p ~\mid~ \mm\tolimp p\}$.
\end{defi}
Note that in the definition above,  $p$ (item 1.) is a scalar, while $\beta$ (item 2.) is subdistribution over   normal forms. The former
  is a quantitative  version of a boolean (yes/no) property,  to reach a normal form. The latter,
  is a quantitative (more precisely, \emph{probabilistic}) version of  ``which normal form is reached.''

Clearly
\begin{center}
	$\pLim(\mm)=\{ \norm {\beta} ~\mid~  \beta\in \Lim(\mm) \}$
\end{center}
because
$ \sup_n  \nnorm{\mm_n}    = \norm{ \sup_n   {\nf \mm}  }  $ (by Lemma~\ref{lem:MCT}, point 1.).

A computationally natural question  is if  the result of computing an element $\mm$ is well defined.
We analyze it  in \refsec{QARS}---putting  this question in a more general, but also simpler,  context.
In fact,
most  properties of the asymptotic behaviours of PARSs are not specific to probability, and are best understood when focusing only on the essentials, abstracting from the details of the formalism.
Before doing so, we build an intuition by informally investigating  the notions of normalization, termination, and  unique normal form in our  concrete setting.

\subsection{PARS vs ARS:\texorpdfstring{\@}{} Subtleties, Questions, and Issues}
\subsubsection{On Normalization and Termination}\label{sec:norm}
In the setting  of ARS, a rewrite sequence from an element $c$ may or may not reach  a normal form. The
  notion of reaching a normal form   comes in  two flavours  (see  Section~\ref{sec:ARS}):
(1.) \emph{there exists}  a rewrite sequence from $c$  which leads to a normal form (\emph{normalization},  \WN); (2.)
	 \emph{each} rewrite sequence from $c$ leads to a normal form (\emph{termination}, \SN).
If no   rewrite sequence  leads to a normal form, then $c$ \emph{diverges}.

It is interesting to  analyze  a  similar $\exists/\forall$ distinction in a  quantitative setting.
We distinguish two cases.

\paragraph{Convergence with  \emph{probability 1}.}
If  we restrict the notion of convergence to \emph{probability 1}, then it is natural to say that an element $\mm$ \textbf{weakly normalizes} if it has a rewrite sequence which
converges with probability $1$, and\textbf{ strongly normalizes} (or, it is \AST) if all  rewrite sequences
converge with probability $1$.

\paragraph{The general case.}
 	Many natural examples---in particular when  we consider untyped probabilistic $\lam$-calculus---are not limited to convergence with probability $1$, as  Example~\ref{ex:motivation2} shows.
 In the general case,   extra subtleties emerge, due to the fact that
\emph{each rewrite  sequence   converges  with some probability} $p\in [0,1]$ (possibly 0).

 A  first important observation is that the set {$\{q\mid \mm\tolimp q\}$} has a supremum (say $p$), but
 \emph{not necessarily  a greatest element}. Think of $\ropen{0,p}$, which has a sup, but not greatest element. If $\pLim$ has no greatest element, it means that  no rewrite sequence converges to the supremum $p$.

A second remark is that we naturally speak of termination/normalization with probability $0$. Not only does it appear  awkward   to separate the case  $0$ (as distinct from $0.00001$), but divergence also---dually---should  be quantitative.

We say that $\mm$ \textbf{(weakly) normalizes} (with probability $p$) if $\{q\mid \mm\tolimp q\}$ has a greatest element $p$. This means that there \emph{exists} a reduction sequence whose limit is $p$.
 Dually, we can say that
$\mm$ \textbf{{strongly normalizes} (or terminates) } (with probability $p$), if all reduction sequences converge with  the same probability $p\in [0,1]$.

Since in this case  \emph{all reduction sequences  from the same element have the same behaviour},  a better term seems   that $\mm$  \textbf{uniformly normalizes}. And indeed, ``all reduction sequences from the same element converge with  the same probability'' is the analogue  of
 the ARS notion of \emph{uniform normalization},  the property that    all reduction sequences from an element   either all terminate, or all diverge (otherwise stated: weak normalization implies strong normalization).
Summing up, we  use the following terminology:

 \begin{defi}[Normalization and Termination]\label{def:SNlim} A PARS is \WNlim, \SNlim,  or \AST, if each $\mm$ satisfies the corresponding  property, where
 	\begin{itemize}

 	 \item $\mm $ is \WNlim ($\mm$ \textbf{normalizes})  if there\emph{ exists a sequence} from $\mm$ which converges with  greatest probability (say $p$). To specify, we say that  $\mm$ is  p-\WNlim.

 		\item $\mm$  is  \SNlim  ($\mm$ \textbf{strongly---or uniformely---normalizes}) if \emph{all sequences}  from $\mm$ converge with the same probability (say $p$).  To specify, we say  that  $\mm$ is  p-\SNlim.

\item   	$\mm$ is  \textbf{Almost Sure Terminating (\AST)} if it strongly normalizes with probability $1$ (\ie, it is 1-\SNlim).

 	\end{itemize}
\end{defi}
\begin{exa} The system in Fig.~\ref{fig:dyadic} is $1$-\WNlim, but not $1$-\SNlim.  The top  rewrite sequence (in blue) converges with probability {$1=\lim_{n\to\infty} \sum_{k:1}^n \frac{1}{2^k}$.} The bottom rewrite sequence (in red) converges with probability  $0$. In between, we have all dyadic possibilities. In contrast, the system in Fig.~\ref{fig:continuum}
	is \AST\@.
\end{exa}

\subsubsection{On  Unique Normal Forms and Confluence}\label{sec:unique}
\newcommand{\ULD}{\texttt{ULD}}
%The main question on which we focus   is when the notion of result $\den \mm$ is well defined.
We now focus on two natural questions. First: when is the notion of the result $\den \mm$  well defined? Second:
%\RED{how   rewrite sequences from the same initial $\mm$ compare?} %More precisely,
given  a probabilistic program $M$,  if
   $[M]\tolim \alpha$ and  $[M] \tolim \beta$,   how do $\beta$ and $\alpha$ relate?

 Normalization and termination are  \emph{quantitative yes/no} properties---we are only interested in the number $\norm{\beta}$,  for $\beta$
 limit distribution; for example, if $\mm\tolim \{\false^1\}$ and $\mm\tolim \{\true^{1/2},\false^{1/2}\}$, then $\mm$ converges with probability $1$,
 but we make no distinction between the two---very different---results. Similarly, consider again  Fig.~\ref{fig:continuum}. The system is \AST,
 however  the limit distributions are\emph{ not unique}: they   span  an  infinity of distributions which have shape $ \{\true^p,\false^{1-p}\}$.	  These observations motivate  attention to finer-grained properties.

In the usual  theory of rewriting,  the fact that the result is well defined is expressed by the   \emph{unique normal form} property (\UN). Let us  examine an analogue of \UN\ in a  probabilistic setting. An intuitive candidate  is the following, which was first proposed in~\cite{DiazMartinez17}:
\begin{center}
	 \ULD:\@ if  $\alpha,\beta\in \Lim(\mm)$, then $\alpha = \beta$
\end{center}~\cite{DiazMartinez17}  shows that,  in the case of \AST,   confluence implies  \ULD\@.
 However, \ULD\ is not a good analogue in  general, because  a PARS does not need to be \AST\ (or \SNlim); it may  well be that $\mm\tolim \alpha$ and $\mm\tolim \beta $, with $\norm \alpha \not= \norm \beta$. We have seen rewrite systems which are not \AST\ in  Fig.~\ref{fig:dyadic},  and  in Example~\ref{ex:motivation2}.   Similar examples are natural in an \emph{untyped} probabilistic $\lambda$-calculus (recall that the $\lambda$-calculus is not \SN!).

 We then prefer not to  limit the analysis to $\AST$. In such  a  case, \ULD\  is not implied by confluence: the system in  Fig.~\ref{fig:dyadic} is indeed confluent, but not \ULD\@. Still, we  would like to say that it satisfies a form of \UN\@.

 We  propose as  probabilistic analogue of \UN\ the following property
 \begin{center}
 	\UNlim:  $\Lim (\mm)$ has a \emph{greatest} element.
 \end{center}
which we justify in  \refsec{QARS}, where
we show that PARS satisfy  an analogue of
standard  ARS results:  ``Confluence implies \UN'' (\refthm{confluence}),
{ and ``the Normal Form Property implies \UN'' (\refprop{NFP}).}
There are however two important observations to make.

\paragraph{Important observation!}
While the statements are similar to the classical ones, the
content is not.   To understand the difference, and what is  non-trivial here,  observe that in general   there is no   reason  to believe that $\Lim (\mm)$ has maximal elements. Think again of the set $\ropen{0,1}$, which has no max, even if it has a sup. Observe  also that $\Lim (\mm)$ is---in general---uncountable.

In Section~\ref{sec:confluenceQ} we will see that   to prove the existence of maximal limits is indeed not immediate. For this reason, while in the case of  finitary termination  uniqueness of normal forms  follows immediately from confluence, it is not so when termination is asymptotic: confluence does not directly guarantee \UNlim, and more work is needed.

\paragraph{Which notion of Confluence?}
Property \UNlim is guaranteed by   a   form of confluence weaker  than one would expect.
	Assume  $\tms\leftstar\leftleftarrows\mm\redd^*\tmr$; with the standard notion of confluence in mind, we may  require that    $\exists \tmu$ such that
	$\tms\redd^* \tmu$,   $ \tmr\redd^* \tmu$ or   that  $\exists \tmu,\tmu'$ such
	that
	$\tms\redd^* \tmu$,   $ \tmr\redd^* \tmu'$ and  $\eqflat \tmu {\tmu'}$. Both are fine, but in \refsec{confluenceQ} we show that  a weaker notion of equivalence (which was already discovered in~\cite{AriolaBlom02})
	suffices---we only need to compare multidistributions w.r.t.\ their information content on normal forms.

	\begin{rem} In the case of \AST (and \SNlim),   all limits are maximal, hence \UNlim becomes \ULD\@.
	%			 if $\alpha,\beta\in \Lim(\mm)$ then  $\alpha =\beta$ (\ULD).
\end{rem}

%%%%%%%%%%%%%%%%%%%%%%%%

\subsubsection{Newman's Lemma Failure, and Proof Techniques for PARS}\label{sec:proof_techniques}\label{proof_techniques}
The statement of \refthm{confluence} ``Confluence implies \UNlim''  has the \emph{same flavour} as the analogue  one for ARSs, but the  \emph{notions are not the same}. The notion of limit  (and therefore that of  \UNlim, \SNlim, and \WNlim) does not belong to the theory of ARSs.  For this reason, the rewrite system $(mA,\redd)$ which we are studying is not simply an ARS\@.
One should not assume that standard  properties of ARSs  transfer to their asymptotic analog.
An illustration %of both issues
of this is   \textbf{Newman's  Lemma}.
Given a PARS, let us assume \AST\  and observe that  in this case,  confluence \emph{at the limit} can be identified with  \UNlim.
\emph{A wrong attempt:} 	\AST\ + \WCRlim $\Rightarrow$  \UNlim,
where  \textbf{\WCRlim:}
if  $\mm \redd \tms_1$ and $\mm \redd \tms_2$, then $\exists\tmr$, with  $\tms_1 \tolim \tmr$, $\tms_2 \tolim \tmr$.
This does not hold. A counterexample is the PARS in Fig.~\ref{fig:continuum}, which does satisfy \WCRlim.

	\begin{rem}
		% Newman's Lemma in fact illustrates well also all these aspects.
	\emph{	Could \emph{a different formulation} uncover properties similar to Newman Lemma?
		Another ``\emph{candidate}'' statement  we can attempt  is : \AST + WCR $\Rightarrow$ \UNlim.  Unfortunately, here we did not find an  answer. However, this property is   an interesting  case study.  It is not hard to show that such a property holds % chktex 26
		when  $\Lim(\mm)$ is finite, or uniformly discrete,  meaning that---given a  definition of distance---there exist a minimal distance between two elements in $\Lim(\mm)$.  This fact  also implies that   a counterexample (if any) cannot be trivial. On the other side, if the property holds,
		the difficulty is  which proof technique to use, since well-founded induction is not available to us.}
	\end{rem}

What  is at play here is that the notion of \emph{termination} is not the same for  ARSs and for PARSs.
A fundamental fact of ARSs (on which all proofs of Newman's Lemma rely) is:  termination implies that   the rewriting relation   is  well founded.   All terminating ARSs allow well-founded induction as proof technique; this is\emph{ not the case} for probabilistic termination.
To transfer properties from ARSs to PARSs there are two   issues: we need to find the \emph{right formulation} and the \emph{right proof technique}.

Notice  that our counter-example above still leaves open the question
``Can a different formulation uncover properties similar to Newman's Lemma?''
Or, better,
``Are there  \emph{local properties} which guarantee \UNlim?''

\section{Quantitative Abstract Rewriting Systems}\label{sec:QARS}\label{sec:qars}

We observed that the notion of result as a limit does not belong to ARSs. However, in many arguments we do not need all the structure coming from PARS\@.
To be able to study   asymptotic rewriting, in this section we define Quantitative Abstract Rewriting Systems (QARS). As already noted, QARS are  a natural refinement of ARSI in~\cite{AriolaBlom02}---we simply move from partial orders to $\omega$-cpos.  The main contribution of this section is to provide a set of proof techniques, first to study properties of the limits, and then to compare reduction strategies.
Working abstractly allows us to study the  asymptotic  properties, capturing   the essence of the  arguments.

\paragraph{QARS} We can see computation as  a process that produces a result by gradually increasing the amount of  available information.
So a reduction sequence gradually computes a result by  converging (in a finite or infinite number of steps) to the maximal amount of information which it can produce.
The standard structure to express a result in terms of partial information is that of an $\omega$-cpo.

Recall that a partially ordered set $\Set=(\Set,\leq)$  is
an \textbf{$\omega$-complete partial order}  (\textbf{$\omega$-cpo}) if  every $\omega$-chain $\cpos_0\leq \cpos_1\leq \ldots$ has a supremum in $\Set$. We assume the partial order to  have  a least element $\bot$.
We denote the elements of $\Set$ with bold letters $\cpor,\cpos, \cpou$\ldots

Let $(\AA,\red)$ be an ARS\@. To each element
$\tmt\in \AA$
 it is associated  a notion of  (partial) information, which is modeled  by  a function from $\AA$  to an $\omega$-cpo.  Def.~\ref{def:qars} formalizes this intuition.
\begin{defi}[QARS]\label{def:qars}
	A Quantitative  ARS  (QARS)  is an ARS $(\AA,\red)$ together with a function  $\obs: \AA \red \Set$, where  $\Set$ is  an $\omega$-cpo and such that
	\begin{center}
		$\tm \red \tm'$ implies  $\obs (\tm)\leq \obs (\tm')$.
	\end{center}
%	The notion  generalizes to a family of functions $\obs_\alpha: A \red \Set_\alpha$.
\end{defi}
Intuitively, the   function $\obs$ observes a specific property of interest about $\tmt \in \AA$. The observation
$\obs(\tmt)$  indicates how much stable information $\tmt$ delivers: the information content is monotone increasing during computation.
%A normal form, if any is reached, should carry  a maximal amount of information.

\newcommand{\ch}[2]{\obs_{#1}(#2)}%{\chi_{#1}(#2)} % chktex 36
\begin{exa}\label{ex:nfs}\label{ex:cpo_num} The following are  examples of QARS\@.
	\begin{enumerate}
		\item\label{p:ARSnf} ARSs: take $\Set=(\{0,1\},\leq)$ and the boolean function $\w \tmt =1$ if $\tmt$ is a normal form, $\w \tmt =0$ otherwise.

		\item  PARS:\@  take  $\Set=([0,1],\leq)$, and a function   which corresponds to a  probability measure, for example the probability to be in normal form $\w {\mm}\mapsto \nnorm \mm$ (as defined in \refsec{observations}).

	\end{enumerate}
\end{exa}

\noindent
The  observation $\obs(\tmt)$ does not need to take numerical values.    In the examples below,  $\Set$ is   an $\omega$-cpo of \emph{partial results}.
\begin{exa}\label{ex:cpo_nf}
	\begin{enumerate}
		\item ARS:\@ take for $\Set$ the flat order on normal forms, and define  the function $\w  \tmt =  \tmt$ if $\tmu$ is normal,
		$\w \tmt =  \bot $ otherwise.

		\item PARS:\@  take  for $\Set$ the $\omega$-cpo of subdistributions on normal forms $\DST{\NFA}$,
		and for $\obs$ the function $\nf -$, as defined in \refsec{observations}.

		%	\item Infinitary $\lam$-calculus: take the  $\omega$-cpo of the partial normal forms which are associated to   $\lam$-terms
		%(see~\cite{AmadioCurien}, pag. 52).
	\end{enumerate}
\end{exa}

\paragraph{Maximal rewrite sequences.}

	From now on, to indicate in a uniform way \emph{maximal} rewrite sequences, \emph{whenever finite or infinite, }
	we write  $\seq \tmt $   for an   {infinite} sequence such that either
	$\tmt_i\red \tmt_{i+1}$, or $\tmt_i=\tmt_{i+1} \in \NFA$ (hence,   $\seq \tmt$ is constant from an index $k$ on, \ie  $\tmt \red^* \tmt_k\not \red$).
Letters $\mathfrak{s,t}$ range over maximal  sequences.

We still write $\tmt\red^*\tms$ to indicate that there is a \emph{finite} sequence from $\tmt$ to $\tms$.

\subsection{Limits as   Results}\label{sec:limits}
In this section we let $\QQ=( (\AA,\red),\obs )$ be an arbitrary but fixed QARS\@.
Intuitively,
the result computed by a possibly infinite reduction  sequence $\seq \mm$ is the  \emph{limit} observation.

By definition, given a $\red$-sequence $\seq \mm$, its limit w.r.t. $\obs$
\begin{center}
	$ \sup_{n} \{\w{\mm_n} \} $.
\end{center}
always exists,  because    $\Set$ is an $\omega$-cpo.
{Intuitively,
this is the maximal amount of information produced by the  sequence,    the \emph{result} of that specific computation.}

If $\red$ is a deterministic reduction---and so from $\tm$   there is a unique maximal $\red$-sequence---it   is standard  to interpret  the   limit as the \emph{meaning} of   $\tmt$.
%, however this definition  makes sense only when $\red$ is a deterministic reduction.
In a QARS, however,  $\tmt$ has \emph{several possible rewrite  sequences}, and therefore can produce  several results/have several  limits.

%%%%%%%%%%%%%%%%

\begin{defi}[$\obs$-limits] For $\mm\in \AA$, we write
	%	\item  	if  a sequence $\seqt$ from $\tm$  such that  $\wlim\seqt = \cpor$ exists,  we  write
	%	\begin{center}
	%		$t\red\conv{\obs} \bm{p}$
	%	\end{center}
	%
	\[\mm \red\conv{\obs} \cpor\]
	if there exists  a sequence $\seq \mm$ from $\mm$  such that  $\sup_{n} \{\w{\mm_n} \} = \cpor$. Then
	\begin{itemize}
		\item  $\wLim (\mm) :=\{\cpor\mid \mm \red\conv{\obs} \cpor\}$ % is  the set of limits of $t$ (w.r.t. $\obs$).

		\item $\den \mm$ denotes the greatest element of  $\wLim (\mm)$, if any exists.

	\end{itemize}

\end{defi}

\noindent
Informally, to  $\mm$ is associated a well-defined result, which we denote $\sem \mm$, if the maximal amount of information produced by any  reduction sequence  is well defined.
The intuition is that $\den \mm$ is well defined if different reduction sequences from $\mm$ do not produce ``essentially different'' results: if  $\cpos \not =  \cpos'$ then they are both \emph{approximants} of a  same result $\cpor$
(\ie, $\cpos, \cpos'   \leq  \cpor$).

Thinking of usual rewriting, consider  $\obs$ as defined in point 1, \refex{cpo_nf}. Then  to have a greatest limit  exactly corresponds to uniqueness of normal forms.

\begin{exa}\label{ex:convergence}Let us revisit \refex{cpo_num}.
	\begin{enumerate}
		\item ARS:\@	consider usual $\lam$-calculus with $\beta$-reduction. The term $\tm = (\lam x. z)(\Delta\Delta)$ has infinitely many  possible  $\red_{\beta}$-sequences.
	With the same definition of $\obs$ as in  \refex{cpo_num}, Point 1,	the set of  limits w.r.t. $\obs$ contains two elements: $\xLim{\obs}(\tm)= \{0, 1\}$.

		\item PARS (probabilistic  $\lam$-calculus): consider   $\m=\mset{1\, I\oplus \Delta\Delta}$, which  has   exactely one  maximal reduction sequence, starting with  $\m \Red 	\mdist{
			\two I, \two \Delta\Delta}\Red\dots$. Define $\obs(\m)=\nnorm \m$. In this case  $\m\Red\conv \obs \two$ and $\wLim(\m)=\{\two\}$.

	\end{enumerate}

\end{exa}

\noindent
Point 2.\ in \refex{convergence} shows well that the notion of result is \emph{quantitative}: $\m$ reaches a normal form with probability $\two$.
This also shows that  maximal elements of $\wLim(\m)$ do not need to be maximal elements of $\Set$; the reason for this choice is exactly that  terms like
$I\oplus \Delta\Delta$ (which converges with probability $\two$ rather than $1$) are natural in an untyped setting like $\lam$-calculus.
As a consequence, in general,  \emph{the  set of limits may or may not have maximal elements}.
	Note that, even if  $\wLim(\m)$  has maximal  elements,  a greatest limit  does not necessarily exist. The probabilistic  $\lam$-term  in \refex{motivation}  is a good example:  different reduction sequences lead to different limits.  Another clear example is \refex{continuum}: $a$ has an infinity of limits, all maximal.

	\begin{rem}[greatest limit]
\emph{We are interested  in  the case when \emph{a greatest limit} does exist. {The reason  is that if $\wLim(\m)$  has a sup $\cpos \in \Set$ which does not belong to $\wLim(\m)$, no rewrite sequence converges to $\cpos$. That is, we cannot compute $\cpos$ {internally in the calculus}}.}
	\end{rem}

	Like for PARS, the natural question is if  the result of computing an element $\mm$ is well defined. This is exactly the sense of the property \UNlim, which we can state in full generality  for QARS\@.
	\begin{center}
		A QARS satisfies property \UNlim~if~$\wLim (\mm)$ has a \emph{greatest} element.
	\end{center}
	Clearly, by definition:
	\begin{center}
		\UNlim\ \quad 	if and only if  \quad $\den \mm$ is defined.
	\end{center}

	\subsection{Confluence and \texorpdfstring{\UNlim}{UNlim}}\label{sec:GL}\label{sec:confluenceQ}

	In our   setting, maximal limits play a  role similar to  that of normal forms in ARSs. However, since termination is asymptotic, the situation is more complex than in a finitary case.
Notably,  in the case of QARS, \emph{confluence does not guarantee   \UNlim}, at least not in general.  In this section, we show that  confluence (and variants of it) imply the following:  if  a  maximal element exists in $\wLim(\m)$, it is the \emph{greatest} element. Note that such a property is stronger than uniqueness of maximal limits---however, it does not imply  \UNlim,   because  we have no guarantee that $\wLim(\m)$ contains any maximal element.

Fortunately, in the case of PARS, confluence \emph{does}  imply \UNlim. However, the proof (Section~\ref{sec:PARS_UN}) relies on
	more properties than the basic ones   which we have assumed for QARS\@.

	\begin{defi}[Confluence]\label{def:confluence}
		A QARS $((\AA,\red),\obs)$ satisfies
		\begin{itemize}
			\item 	Confluence if: for all
			$ \tms,\tmr \in \AA$ with  $\tms\leftstar\leftarrow\mm\red^*\tmr$, there exist   $ \tmu$ such
			that
			$\tms\red^* \tmu$,   $ \tmr\red^* \tmu$.

			\item  $ \obs $-Confluence if: for all
			$ \tms,\tmr \in \AA$ with  $\tms\leftstar\leftarrow\mm\red^*\tmr$, there exist   $  \tms',\tmr'$ such
			that
			$\tms\red^* {\tms'}$,   $ \tmr\red^* \tmr'$, and $\w{\tms'} = \w{\tmr'}$.

			\item Skew-Confluence\footnote{In~\cite{Faggian19}   we call this property  Semi-Confluence. The same property is studied in~\cite{AriolaBlom02}---here we adopt their terminology.} if:
			$\forall \tms,\tmr \in \AA$ with  $\tms\leftstar\leftarrow\mm\red^*\tmr$,  exists  $ \tms'$ such
			that
			$\tms\red^* \tms'$,  and $\w\tmr \leq  \w{\tms'}$.

		\end{itemize}
	\end{defi}
    \noindent
	Clearly,
	\begin{fact}\label{fact:confluence}
		Confluence $\implies$ $ \obs $-Confluence  $\implies$ Skew-Confluence.
	\end{fact}

	In analogy to the normal form property of ARS (\NFP, see \refsec{ARS}), we define the following

	\begin{defi}[Limit Property (\LimP)]\label{def:LimP} A QARS $((\AA,\red),\obs)$ satisfies  the Limit Property   \LimP  if ($\forall \mm,\tms \in \AA$):

		%	\item \textbf{\NFPlim\ (Normal Form Property):} if $\alpha$ is maximal in $\wLim(\mm)$, and  $\mm\red^*\tms$ then $\tms\tocpo \alpha$.

		\begin{center}
			$\cpor\in \wLim(\mm)$ and  $\mm\red^*\tms$ imply that  there exists $\cpos\in \wLim(\mm)$ such that
			$\tms \tocpo \cpos$ and $\cpor\leq \cpos$.
		\end{center}

	\end{defi}

	\begin{center}
		\begin{figure}\centering
			\fbox{		\includegraphics[width=0.5\textwidth]{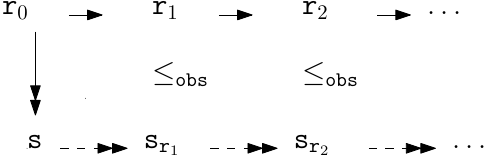}}
			\captionsetup{width=\linewidth}  \caption{Skew-Confluence implies  \LimP}\label{fig:weak_conf}
		\end{figure}
	\end{center}

	\begin{lem}[Main Lemma]\label{l:main}Given a QARS,
		Skew-Confluence implies    \LimP.
	\end{lem}
	\begin{proof}
	Let $\tmr_0, \tms\in\AA$, ${\cpor}\in \wLim(\tmr_0)$,  $\tmr_0\red^*\tms$.
		  Let  $\seq \tmr$ be a sequence with  limit  ${\cpor}$.
		As illustrated in Fig.~\ref{fig:weak_conf},
		starting from $\tms$, we build a sequence  $\tms=\tms_{\tmr_0}\red^* \tms_{\tmr_1} \red^*\tms_{\tmr_2}\dots$, where $\tms_{\tmr_i}$, $i \geq 1$ is given by Skew Confluence: from   $\tmr_0\red^* \tmr_i$  and
		$\tmr_0\red^*\tms_{\tmr_{i-1}}$
		we  obtain
		$\tms_{\tmr_{i-1}}\red^* \tms_{\tmr_{i}}$  with
		$\w{\tmr_{i} } \leq  \w{\tms_{\tmr_{i}}}$.
		Let ${\cpos} $ be the limit of the  sequence so obtained; observe that   $ {\cpos}\in \wLim(\tmr_0)$.  By  construction,    $\forall i$,
		$\w{\tmr_i}\leq  \w{\tms_{\tmr_{i}}} \leq {\cpos} $. From ${\cpor}=\sup{\langle{\w{\tmr_n}}\rangle}$ it follows that  ${\cpor}\leq {\cpos}$.
	\end{proof}

	\LimP\ implies that  if a  maximal limit exists, it is the \emph{greatest} limit.

	\begin{prop}[Greatest  limit]\label{prop:unique}\label{prop:NFP} Given a QARS $((\AA,\red),\obs)$, and $\mm\in \AA$,
		 \LimP\ implies that if  $\wLim(\mm)$ has a maximal element, then
		it is the greatest element.
	\end{prop}
	\begin{proof}Let ${\cpor} \in \wLim(\mm)$ be maximal. For each  ${\cpou}\in \wLim(\mm)$, there is a sequence $\seq \mm$ from $\mm$ such that ${\cpou} = \sup_n \w{\mm_n}$.
		\LimP\   implies that  $\forall n$,  $\mm_n\toinf {\cpos}_n \geq {\cpor}$. By maximality of ${\cpor}$,  $ {\cpos}_n = {\cpor}$ and therefore $\w{\mm_n}\leq {\cpor}$.
		From ${\cpou}=\sup_n \w{\mm_n}$ we conclude that  ${\cpou}\leq {\cpor}$, that is, ${\cpor} $ is \emph{the greatest element} of
		$\wLim(\mm)$.
	\end{proof}

	Given a confluent QARS, to  guarantee that \UNlim holds,  and therefore for each $\mm\in \AA$,  $\sem \mm$ is defined, it suffices to establish that $\wLim(\mm)$ has a maximal element.

In \refsec{PARS_UN}	we prove that in the case of PARS, confluence implies  the existence of a maximal element and therefore of a  greatest element.
	To do that, we use more structure, namely the fact that the $\omega$-cpo $\DSTNF$ is equipped with an order-preserving norm $\norm{}: \DSTNF \to [0,1] $.

\subsection{Observing in the unit interval}\label{sec:bQARS}
	\renewcommand{\weight}{\obs}
	Let us consider the  case  of QARS where the associated  $\omega$-cpo is the \emph{bounded} interval $[0,1]\subset \Real$,  equipped with the standard order.

	Assume fixed  a QARS $\Q=((\AA,\red), \weight)$ such that   $\weight: \AA\to [0,1]\subset \Real$.
	We show that  property \LimP\  implies that $p= \sup\, \wwLim(\mm)  $ \emph{belongs} to  $ \wwLim(\mm)  $, where
	$\wwLim(\mm) = \{ q ~\mid~ \mm \tolimww{\red} q  \} $. Therefore, \UNlim holds.

	We need a technical lemma
	\begin{lem}\label{l:Lmain} Let $\Q$,  $\Norms (\mm)$ and $p $ be as above. For each $\epsilon >0$, property \LimP\ implies the following:
		if $ q\in \Norms (\mm)$, $|p-q|\leq \epsilon$, and $\mm \red^* \tms$, then  there exists $\tmt$, such that $\ms \red^* \tmt $ and $|p-\ww{\tmt}| \leq 2\epsilon $.
	\end{lem}
	\begin{proof}The assumption $\mm \red^* \tms$ and
		\LimP imply that there exists a rewrite sequence $\seq \tms$ from $\tms$ which converges to  $q'\geq q$; clearly $|p-q'|\leq \epsilon$.

		By definition of limit of a sequence,   there is an index $k_{\epsilon}$    such that
		$|q'-\ww {\tms_{k_{\epsilon}}}|\leq \epsilon $, hence $|p-\ww {\ms_{k_{\epsilon}}}|\leq 2\epsilon$. Since $\ms \red^* \ms_{k_{\epsilon}}$, $\tmt=\ms_{k_{\epsilon}}$ satisfies the claim.
	\end{proof}

	\begin{prop}[Greatest limit]\label{prop:bool_UN} Given a  QARS $\Q=((\AA,\red), \weight)$ such that   $\weight: \AA\to [0,1]$, property \LimP\ implies that  $\wwLim(\mm) $ has a greatest element.
	\end{prop}

	\begin{proof}
		Let $p =\sup~\wwLim(\mm)$. We show that   $p\in \wwLim(\mm)$,
		by building  a rewrite  sequence $\seq \mm$ from $\mm$ such that  $\seq \mm \tolim  p$.

			For each $k\in \Nat$, we define   $\epsilon_k=\frac{1}{2^k}$.
			Observe that  for each $\epsilon\in \Real^+$, there exists   $q\in \Norms  $ such that   $ |p-q| \leq \epsilon$.

	Let $\tms_{0}=\mm$. From here, we build a sequence of reductions
		$\mm\red^* \tms_{1}\red^*\tms_{2}\red^*\dots$ whose limit is  $p$, as follows. %illustrated in Fig.~\ref{fig:exists}.
		For each $k>0$:
		\begin{itemize}
			\item
			there exists {$q_k\in \Norms(\mm)$} such that   $ |p-q_k| \leq \frac{\epsilon_k }{2} $% = \frac{1}{2} \frac{1}{2^k} $.

			\item  From $\mm\red^* \tms_{k-1}$, we use \reflem{Lmain} to establish that there exists  $\tms_{k}$ such that $\tms_{k-1}\red^* \tms_{k}$ and
			$|p-\ww{{\tms_{k}}}| \leq \epsilon_k$   ($p, ~q_k, ~\tms_{k-1}, ~\tms_{k} $ resp.\ instantiate $p, q, ~\tms, ~\tmt$ of \reflem{Lmain}).
		\end{itemize}
        \noindent
		Let $\seq \mm$ be the concatenation of all the  finite sequences $\tms_{k-1}\red^* \tms_{k}$. By construction,  $\lim_{n\to \infty} \langle \ww{\mm_n} \rangle =p$. We conclude that   $p\in  \Norms(\mm)$.
	\end{proof}

	\subsection{PARS:\texorpdfstring{\@}{} Confluence implies \texorpdfstring{\UNlim}{UNlim}}\label{sec:PARS_UN}
	We now can show that in the case of PARSs, confluence (in all its variants) implies \UNlim (\refthm{confluence}) and therefore for each $\mm$, $ \sem \mm $ is defined.
	In this section, we fix a PARS $(\MA, \redd)$, and define  $\PP=((\MA, \redd), \nf -)$  and $\PPb=((\MA, \redd), \nnorm{-})$,
	where $\nf{-}$ and $\nnorm{-}$ are as  defined in  Section~\ref{sec:observations}.
	It is immediate to check that
	\begin{fact}$\PP=((\MA, \redd), \nf -)$  and $\PPb=((\MA, \redd), \nnorm{-})$ are QARS\@.
	\end{fact}

	Recall that $ \nnorm{-} $ is induced by composing  $\nf{-}$ with the norm     $\norm{~} : \DSTNF \to [0,1]$, and that letters $\alpha,\beta,\gamma$ denote elements in $\DSTNF$.

    \bigskip
	First, we observe that
	\begin{fact}
$\PPb$  satisfies the conditions of  \refprop{bool_UN}. Therefore,  if $\PPb$ satisfies confluence (and so  \LimP), then $\pLim (\mm)$ has a greatest element.
	\end{fact}

	We now  lift the result to $\PP$. Precisely, we  prove that for $\PP$,
	 property \LimP\ (\refdef{LimP})  implies
	\emph{existence  of a maximal element} $\alpha$ of $\Lim(\mm)$. Then  (by \refprop{unique}) $\alpha$ is the greatest element of $\Lim(\mm)$.
We rely on the following properties, which we already established in \refsec{limit}.
	\begin{itemize}
		\item $\alpha <\beta$ implies   $\norm \alpha < \norm \beta$,
		%and therefore ($\mm\redd \mm'$ implies  $\nnorm \mm \leq  \nnorm {\mm'}$);

		\item $\pLim(\mm)=\{ \norm {\beta} ~\mid~  \beta\in \Lim(\mm) \}$
	\end{itemize}

	\begin{lem}
		If $\PP$ satisfies \LimP, then  $\PPb$ also does.  Similarly for  all variants of confluence in \refdef{confluence}.
	\end{lem}

	\begin{prop}[Maximal elements]\label{prop:exist} If $\PP$ satisfies \LimP,  then
		$\Lim(\mm)$ has maximal elements.
	\end{prop}

	\begin{proof}   $\xLim{\nnorm{-}}(\mm) $ has a greatest element.
		We observe that if  $\alpha \in \Lim(\mm) $ and $\norm \alpha$ is maximal in $\xLim{\nnorm{-}}(\mm)$, then $\alpha$ is maximal in $\Lim(\mm)$ (because if  $\gamma \in \Lim(\mm)$ and $\gamma >\alpha$, then $\norm \gamma > \norm \alpha$).
	\end{proof}

	\begin{thm}[Confluence implies \UNlim]\label{thm:confluence}Given a PARS, any variant of confluence in \refdef{confluence}  implies \UNlim.
	\end{thm}

	\begin{proof}
		The claim follows  from Fact~\ref{fact:confluence}, \reflem{main}, and   Propositions~\ref{prop:unique} and~\ref{prop:exist}.
	\end{proof}

	%%%%%%%%%%%%%%%%%%%%

	%
	\section{Tools for the analysis of QARS }

We closed \refsec{proof_techniques} with the question:
\begin{center}
		``Are there  \emph{local properties} which guarantee \UNlim?''
\end{center}
This section  develops     criteria of this kind.

    \bigskip
		If the result $\sem \mm$ of computing $\m$ is well defined, the next  natural question is how to compute it: does there exist a strategy $\reds \ \subseteq \ \red$ whose limit is guaranteed to
	be $\sem \mm$? More generally:  does there  exist a strategy $\reds \ \subseteq \ \red$ whose limit is guaranteed to
	be a maximal element of $\wLim(\m)$, if it exists?

	We introduce some tools to help in this analysis. Our focus is  on  properties which can be expressed by   \emph{local conditions}.

	\subsection{Weighted Random Descent}\label{sec:wRD}\label{sec:RD}
	We present a method  to establish---with a \emph{local} test---that for each  element $\m$ of a QARS, $\wLim(\m)$  contains a \emph{unique element} by  generalizing the ARS property of Random Descent.
	Random Descent is not only an elegant  technique in rewriting, developed in~\cite{Oostrom07,OostromT16}, but adapts well and naturally to  the asymptotic setting.

	%\blue{ Main technical result  is a \emph{local characterization} of the property (Thm~\ref{EB char}), similarly to~\cite{Oostrom07}.}

	\paragraph{Random Descent.}
	A reduction $\red$ has \emph{random descent} (RD)~\cite{Newman} if whenever an element  $t$ has normal form, then all rewrite
	sequences from $t$ lead to it, and all have the same length.
	The best-known property which implies RD, as first observed by Newman~\cite{Newman}, is the following
	\begin{center}
		\emph{RD-diamond}: if  $  s_1 \leftarrow t \rightarrow s_2$ then  either $s_1=s_2$, or $  s_1 \rightarrow u \leftarrow  s_2$ for some $u$.
	\end{center}
	This is only a sufficient condition. Quite surprisingly, Random Descent can be characterized by a local (one-step) property~\cite{Oostrom07}.

	\paragraph{Weighted Random Descent}
	We  generalize Random Descent to  observations. The property  $\obs$-RD states that even though  an element $\mm$ may have different reduction sequences, they are all \emph{indistinguishable} if regarded through the lenses of $\obs$.  That is, if we consider all reduction sequences $\seq \mm$ starting from the same $\mm$, they all induce the same
	$\omega$-chain   $\langle \w{\mm_n}\rangle$.
	Obviously, if  all $\omega$-chains  from $\mm$ are equal, they all have the same limit $ \sup_{n} \{\w{\mm_n} \} $.

	The main  technical result of the section  is a \emph{local characterization} of the property (Thm~\ref{EB char}), similarly to~\cite{Oostrom07}.

	\begin{defi}[Weighted Random Descent]\label{def:RD} The QARS  $((\AA,\red),\obs)$   satisfies the following properties (illustrated in Fig.~\ref{fig:RD})
		if they hold for each $\mm\in \AA$.
		\begin{enumerate}
			\item {$\obs$-RD}:  for each pair of $\red$-sequences   $\seq \tmr$,   $\seq \tms$ from  $\tmt$,  and  for each $n\in \Nat$:  $\w {\tmr_n} = \w {\tms_n}$.
			\item \textbf{local} $\obs$-RD:\@   if  $\tmr \leftarrow \mm\red \tms$,   there exists a pair of sequences $\seq \mr$ from $\mr$ and
			$\seq \ms$ from $\ms$ such that,  for  each $n\in \Nat$, $ \w {\ms_n}= \w{\mr_n}$.

			%		\item \textbf{local} \EB\  (\LEB):   if  $\tmt \leftleftarrows\mm \redd \tms$, then for each $k$ there
			%		exist $\tms_k,\tmt_k$  with $\tms \redd^k \tms_k$, $\tmt \redd^k \tmt_k$, and
			%		$\eq{\tms_{k},\tmt_{k}}$.
			%

		\end{enumerate}
	\end{defi}

	\begin{exa}Let us revisit the Random Descent property of weak  Call-by-Value $\lam$-calculus. Let $I=\lam z.z$;
		the following are two different  $\wred$-sequences from the term $(II)(Ix)$.
		\begin{enumerate}
			\item $(II)(Ix) \wred I(Ix) \wred  Ix \wred x$
			\item $(II)(Ix) \wred II(x) \wred Ix \wred  x$
		\end{enumerate}
		Let $\obs: \Lambda \to \{0,1\}$	be $1$ if the term is a value (\ie, a variable or an abstraction), $0$ otherwise. Seen through the lenses of $\obs$, both sequences appear as
		$ \langle	0,~ 0 ,~ 0,~  1 \rangle $.
	\end{exa}

	\begin{exa}
		In Fig.~\ref{fig:continuum}  $\obs$-RD holds for $\obs=\norm{\nf -}$, and  not for $\obs=\nf{-}$.
	\end{exa}

	\begin{figure}\centering
		\fbox{
			\begin{minipage}[c]{0.35\textwidth}
				\includegraphics[width=0.9\textwidth]{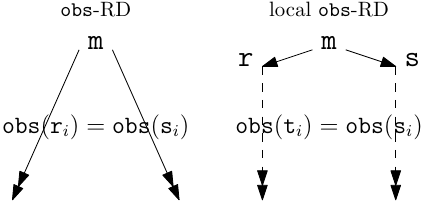}
				\captionsetup{width=\linewidth}  \caption{Random Descent}%
				\label{fig:RD}
			\end{minipage}
		}
		\fbox{
			\begin{minipage}[c]{0.2\textwidth}
				\includegraphics[width=0.8\textwidth]{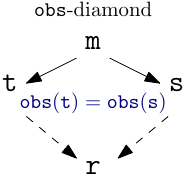}
				\captionsetup{width=\linewidth}  \caption{Diamond}%
				\label{fig:diamond}
			\end{minipage}
		}
		\fbox{
			\begin{minipage}[c]{0.3\textwidth}
				\includegraphics[width=0.8\textwidth]{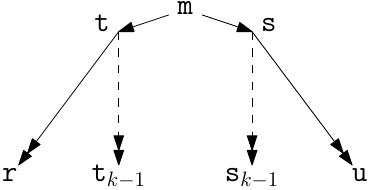}
				\captionsetup{width=\linewidth}  \caption{Proof of Thm.~\ref{RD char}}%
				\label{fig:RDproof}
			\end{minipage}
		}
	\end{figure}

	It is immediate that
	\begin{prop}\label{thm:RD}
		If a QARS $((\AA,\redd),\obs)$ satisfies $\obs$-RD, then for each $\mm\in \AA$,  $\wLim(\mm)$ has a unique element.
	\end{prop}

	While expressive, \EB\ is of little practical use, as it is a property which is \emph{universally quantified} on the  sequences from $\mm$.  %The property   \LEB\ is instead \emph{local}.
	Remarkably,   \emph{the  \LEB\ property
		characterizes}  \EB\@.

	\begin{thm}[Characterization]\label{RD char}\label{EB char}
		The following properties are equivalent:
		\begin{enumerate}
			\item   \LEB; % chktex 13
			\item   $\forall k,  \forall \mm, \tmu, \tmr\in \AA$ if  $\mm\red^k \tmu$ and  $\mm\red^k \tmr$, then $
			\w{\tmu}=\w{\tmr}$;
			\item    \EB\@.
		\end{enumerate}
	\end{thm}
	\begin{proof}The proof is illustrated in  Fig.~\ref{fig:RDproof}.

		$\bm{(1 \implies 2)}$.  We prove that (2) holds  by induction on $k$.
		If $k=0$, the claim  is trivial. If $k>0$, let $\mm\red\tms $ be the first step from $\mm$ to $\tmu$  and
		$\mm\red\tmt$ the first step from	 $\mm$ to $\tmr$.
		By \LEB,
		there exists  $ \tms_{k-1}$ such that  $\tms{\red^{k-1}}\tms_{k-1}$ and  $\tmt_{k-1}$  such that $\tmt\red^{k-1}\tmt_{k-1}$, with
		$\w{\tms_{k-1}}=\w{\tmt_{k-1}}$. Since $\tms \red^{k-1} \tmu$, we can apply the induction hypothesis, and conclude that
		$\w{\tms_{k-1}}=\w{\tmu}$.
		By using the induction hypothesis on $\tmt$, we have that $\w{\tmt_{k-1}}=\w{\tmr}$ and conclude that $\w{\tmr}=\w{\tmu} $.

		$\bm{(2\implies 3)}$. Immediate.

		$\bm{(3 \implies 1)}$. Assume   $\tm \leftarrow\mm \red \tms$. Take a sequence $\seq \tm$ from $\tm$ and a sequence $\seq \tms $ from $\tms$. By (3),
		$\w{\tm_k}=\w{\tms_k}$ $\forall k$.
	\end{proof}

	\paragraph{A diamond.}
	A useful case of \LEB\   is
	the \textbf{$\obs$-diamond}  property (Fig.~\ref{fig:diamond}):
	$\forall \mm,\tms,\tmt$, if  $\tmt \leftarrow\mm \red \tms$, then    $\w{\tms}=\w{\tmt} $, and either $\tms=\tmt$, or
	$\exists \tmu$ s.t.  \big($\tmt\red\tmu\leftarrow  \tms$ \big).
	It is easy to check that 	$\obs$-diamond $\implies$  \LEB\@.

	\begin{prop}\label{prop:diamond}
		($\obs$-diamond) $ \Rightarrow $ \LEB  $ \Rightarrow $  $\wLim(\tmt)$ contains a unique element.
	\end{prop}
	Notice that while \LEB\ \emph{characterizes} \EB, $\obs$-diamond is only a \emph{sufficient condition}.

    \bigskip
	\begin{rem}[The beauty of local]
	\emph{Observe locality at work  in this and next section.   To show  that a  property $P$ holds globally (\ie\ \emph{for each} two rewrite sequences, $P$ holds), we  show  that $P$ holds  locally (\ie\ for each  pair of one-step reductions, \emph{there exist}
		two rewrite  sequences  such that $P$ holds). The space of search when testing  property $ P $ is then reduced, a fact that   we  exploit
		in the proofs of  Section~\ref{sec:weak}.}

	\end{rem}

		\subsection{ Strategies and Completeness}

	\emph{Strategies} are a  way to control the
	non-determinism which arises from  different possible choices of reduction.

	\paragraph{ARS  Strategies.}
	Given an ARS $(\AA,\red)$, a  \emph{strategy} for $\red$ is a relation
	${\ered}\subseteq {\red}$ with  the same normal forms as $\red$.

	\paragraph{QARS  Strategies.}
	Given a QARS $\QQ=(\AA,\red,\obs)$, we call  \emph{strategy}\footnote{Note that the ARS condition of having the same normal forms, is replaced  by the fact
		that   we (tacitly)  consider the QARS   $\big((\AA, \reds), \obs \big)$, where   the function
		$\obs$ is the \emph{same} as for $\QQ$.} for $\red$
	a relation   $\reds\subseteq \red$. We indicate strategies for $\red$ by colored arrows $\reds,\redr$.

	\paragraph{Completeness.} We formulate an asymptotic   notion of  completeness.

	\begin{defi}[Completeness]\label{def:Acompleteness}
		A  reduction $\reds\subseteq \red$ is
		$\obs$-complete (or asymptotically compelte) for $ \red $ if:
\begin{center}
			$\tmt \tolimO{\red} \cpos$ implies  $\tmt \tolimO{\reds} \cpor $ with $ \cpor\geq \cpos $.
\end{center}
	\end{defi}

	Note that a strategy $\reds$ which is asymptotically complete is not guaranteed to find the ``best'' result, as one can immediately see by noticing that   $\red$ is
	trivially a complete strategy for $\red$ itself.
	However, it  does if
	$\reds$ is deterministic, or ``essentially'' deterministic w.r.t. $\obs$.
	This is  what the method in  \refsec{RD} provides.
%	\marginpar{say better...}

	\subsection{Comparing  Strategies}\label{sec:comparing}

	In this section we refine  the results given in the previous section into   a method to compare strategies.
	\emph{In which case is  strategy  $\reds$ \emph{better} than strategy $\redr$?} We  adapt to QARSs  the ARS notion of ``\emph{better}''  introduced in~\cite{Oostrom07}.
	Again,
	we obtain   a local characterization (Thm.~\ref{thm:LD}) of the  property, similarly to~\cite{Oostrom07}.

	In \refsec{comparing_pars} we  will analyze  these notions in the setting of PARS\@. What we obtain  are   sufficient criteria  to establish that a strategy is  normalizing  or perpetual  (Cor.~\ref{cor:normalizing2}),  and to compare the \emph{expected number of steps} of  rewrite sequences.

	The method also	provides another  sufficient condition to establish  \UNlim.

	\begin{figure}%[htbp]
		\centering
		\fbox{
			\includegraphics[width=10cm]{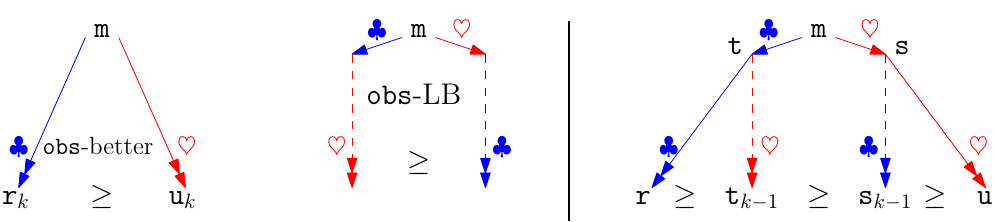}
		}
		\captionsetup{width=\linewidth}  \caption{$\obs$-better}\label{figLB}\label{fig:LB}
	\end{figure}

	\begin{defi}[$\obs$-better]
	The following properties are illustrated in Fig.~\ref{figLB}.
		\begin{itemize}
			\item $\blue{\reds}$ \textbf{is $\obs$-better} than $\todo{\redr}$ ($\B(\blue{\red_{\clubsuit}}, \todo{\redr})$): %\blue{\reds},\todo{\redr}
			for each $\mm$ and for  each pair of a $\blue{\reds}$-sequence $\seq \tmr$  and  a  $\todo{\redr}$-sequence $\seq \tmu$  from $\mm$,
			$\rel{\tmr_n}{\tmu_n}$ holds, for each $n$.

			\item  $\blue{\reds}$ \textbf{is  locally $\obs$-better} than $\todo{\redr}$ (written $\LD(\blue{\red_{\clubsuit}}, \todo{\redr})$):  if $\tmt \blue{\leftclub\leftarrow}\,{\mm}\, \RED{ \redr} \tms$, then for each $n\geq 0$, $\exists \tms_n, \tmt_n$, such that   $\blue{\tms \red_\clubsuit^n \tms_n}$,
			$\todo{\tmt\redr^n ~\tmt_n}$, and  $\rel {\tmt_n}{\tms_n}$
		\end{itemize}
	\end{defi}

	\begin{rem}Please notice that \EB\ (resp. \LEB)   is a special case of $ \B $  (resp. $\LD$).
		We  have  treated \wRD first and independently, for the sake of presentation.
	\end{rem}

	It is immediate that   $\B(\blue{\reds},\red)$ implies that  $\blue{\reds}$  is $\obs$-complete for $\red$ (\refdef{Acompleteness}).
	The notion of $ \B $ is again a  condition which is expressive, but \emph{quantified over all reduction sequences} from $\mm$.
	We now  prove  that the local property   $\LD$
	is sufficient   to establish $\obs$-better, and even necessary when comparing with $\red$.
	\begin{thm}\label{thm:LD}\label{lem:LD}
		$\LD(\blue{\reds},\todo{\redr})$ implies % $ \blue{\reds} $ is $\obs$-better than $ \todo{\redr}$.
		$\B(\blue{\reds},\todo{\redr})$.
		The reverse  also holds if either $\blue{\reds}$ or $\todo{\redr}$ is $\red$.
	\end{thm}

	\begin{proof}
		$\pmb{\Rightarrow}$. The proof is illustrated in Fig.~\ref{figLB}. We prove by induction on $k$ the following:
				\begin{center}
		\emph{  $\LD(\blue{\reds},\todo{\redr})$ implies \\
			$	\big(\forall k, \forall \mm,\tmr,\tmu \in \AA$,   if
			$\blue{\mm \reds^k \tmr}$ and $\todo{\mm \redr^k \tmu}$, then $\rel{\tmr }{ \tmu} \big)$ }.
		\end{center}
		If $k=0$, the claim is trivial. If $k\geq 1$, let $\mm\reds\tms$ be the first step from $\mm$ to $\tmu$, and $\mm\redr\tmt$ the first step from $\mm$ to $\tmr$,
		as in Fig.~\ref{figLB}.  $\LD$ implies that there exist $\tms_{k-1}$ and $\tmt_{k-1}$ such that  $\blue{\tms \reds^{k-1} \tms_{k-1}}$,  $\todo{\tmt \redr^{k-1} ~\tmt_{k-1}}$, with
		$\rel{\tmt_{k-1}}{\tms_{k-1}}$.
		Since $\todo{\tms\redr^{k-1}~ \tmu}$ we can apply the induction hypothesis, and obtain that $\rel{\tms_{k-1}}{\tmu}$.
		Again by  induction hypothesis, from  $\blue{\tmt \reds^{k-1} \tmr}$ we obtain  $\rel{\tmr}{\tmt_{k-1}}$. By transitivity, it holds that $\rel {\tmr}{\tmu}$.

		$\pmb{\Leftarrow}$.  Assume $\todo{\redr}= \red$, and $\tmt \blue{\leftclub\leftarrow}\mm  \red \tms$. Let $\seq{\tmt}$ and $\seq{\tms}$ be obtained by extending $\tmt$ and $\tms$ with a maximal $\blue{\reds}$ sequence. The claim follows from the hypothesis that $\blue{\reds}$ dominates $\red$, by viewing the $\blue{\reds}$ steps in $\seq{\tms}$  as $\red$ steps.
	\end{proof}

	\paragraph{Greatest Element.} Finally,  we mention that  $\LB$
	provides another method  to establish  \UNlim, and therefore the fact that for each $\mm$, $\sem \mm$ is defined.
	\begin{prop}[Greatest Element] Given $\QQ=\{(\AA,\red),\obs\}$,
		if there is a strategy $\blue{\reds}$ such that $\LB(\blue{\reds}, \red)$,  then $\QQ$ satisfies \UNlim.
	\end{prop}
	\begin{proof}
First, observe that the assumption implies in particular $\LB(\blue{\reds}, \reds)$, and therefore (by Thm.~\ref{RD char}) the QARS $ ((\AA,\blue{\reds}),\obs) $ satisfies \wRD. So,  given $\mm\in \AA$, all  $\reds$-sequences from $\mm$ have the same limit $\cpor$. % chktex 13
From   $\LB(\blue{\reds}, \red)$, it follows that $\cpor$ is the greatest limit of $\wLim (\mm)$.
	\end{proof}

\section{PARS:\texorpdfstring{\@}{} Weighted Random Descent}\label{sec:RD_pars}\label{sec:balance}
When applied to PARS,
\wRD  is able to  guarantee some remarkable   properties: \UNlim and  \emph{$p$-\SNlim} as soon as  \emph{there exists a  sequence} which converges with probability $p$, and also the fact that  all rewrite sequences from an element have the same \emph{expected} number of steps.

Take $\obs$ to be  either $\nf -$ or $ \norm{\nf -}$,
   \wRD implies that  all   rewrite sequences from $\mm$:
\begin{itemize}
		\item      have  the same probability of reaching a normal form  after $k$ steps (for each  $k\in \Nat$);
	 	\item   converge to the same limit;
		\item     have the same expected number of steps.
\end{itemize}

\begin{prop}\label{thm:RD_pars}\hfill
	\begin{enumerate}
		\item\label{norm_term} $\norm{\nf {-}}$-RD implies    \SNlim
		\emph{	(\emph{uniform normalization})}; moreover, for each $\mm\in \MA$  all elements in $\Lim (\m)$ are maximal.
		\item   $\nf -$-RD  implies  \SNlim    and   \UNlim.
	\end{enumerate}
\end{prop}

\paragraph{Point-wise formulation.} In  \refsec{weak},   we exploit the fact that not only \EB\  admits a local characterization, but also  that
the  properties \LEB\ and $\obs$-diamond %of  Definition~\ref{def:RD}
can be expressed point-wise, making the  condition  even easier  to verify.
\begin{enumerate}
	\item  pointed \LEB:\@ $\forall a\in A$, if  $\tmt \leftleftarrows [a^1] \redd \tms$, then  $\forall k, \exists
	\tms_k,\tmt_k$  with $\tms \redd^k \tms_k$, $\tmt \redd^k \tmt_k$, and
	$\w{\tms_{k}} = \w{\tmt_{k}}$.

	\item pointed $\obs$-diamond: $\forall a\in A$, if  $\tmt \leftleftarrows [a^1] \redd \tms$, then it holds that  $\w{\tmt}=\w{\tms} $, and
	$\exists \tmr$ such that  $\tmt\redd\tmr \leftleftarrows \tms $.
\end{enumerate}

\begin{prop}[point-wise \LEB]\label{pointwise}
	The following hold
	\begin{itemize}
		\item  \LEB\ $\iff$  pointed \LEB; % chktex 13

		\item $\obs$-diamond $\iff$ pointed $\obs$-diamond.

	\end{itemize}
\end{prop}

\begin{proof}
	Immediate, by the definition of $\redd$. Given $\mm=[p_{i}a_i]_{\iI}$, we establish the result for each $a_i$, and put all the resulting multidistributions together.
\end{proof}

\subsection{Expected Termination Time}\label{sec:meantime}
\newcommand{\PAST}{\texttt {PAST}\xspace}
For ARS, Random Descent   captures  the property    (\texttt{Length}) ``\emph{all maximal rewrite sequences from an element have the same length}.''

\wRD  also implies  a property similar to (\texttt{Length}) for PARS,  where we consider not    the  number of steps of  the rewrite sequences, but    its probabilistic analogue,  the \textbf{expected number of steps}.

In an ARS, if a maximal rewrite sequence terminates, the number of steps is finite; we interpret this number as \emph{time to termination}.  In the  case of PARS, a system may have infinite runs even if it is \AST;\@ the number of rewrite  steps $\red$ from an initial state  is (in general) infinite. However,  what interests us is  its \emph{expected value}, \ie\  the  weighted average w.r.t.\ probability (see Section~\ref{sec:proba}) which we write \MeanTime($\seq \mm$).  This expected value can be finite; in this case, not only the system  is  \AST, but is said \PAST \emph{(Positively \AST)}  (see~\cite{BournezG05}).

	\begin{exa}\label{ex:meantime}

	An example of probabilistic system with finite  expected  time to termination   is the one  in Fig.~\ref{fig:AST}.  The reduction   from $c$ has  \MeanTime\ $2$.
We can see this informally, recalling  Section~\ref{sec:proba}.
Let the sample space  $\Omega$ be the set of paths ending in a normal form, and let $\mu$ be the probability distribution on $\Omega$. What is the expected value  of the random variable
$\len:\Omega\to \Nat$?
We have $E(\len)=\sum_{\omega} \len(\omega) \cdot \mu(\omega)=$
$\sum_{n\in \Nat} n\cdot \mu\{\omega \mid \len(\omega)=n\}=$ $\sum n\cdot \frac{1}{2^n}=2$.
%
%$2$ is the mean time to termination of  $\{a^1\}$.

\end{exa}

%~\cite{Avanzini} makes a nice  observation:  the expected number of steps of a rewrite sequence $\seq \mm $  admits
a very simple formulation, as follows:

	\begin{equation}\label{eq:ET}
	\MeanTime(\seq \mm )= \sum_{n\in \Nat} (1-\nnorm{\mm_{n}})
	\end{equation}

%	\MeanTime($\seq \mm $)= $\sum_{n\in \Nat} (1-\nnorm{\mm_{n}})$

%	\begin{center}
%		\MeanTime($\seq \mm $)= $1 + \sum_{i \geq 1} (1-\nnorm{\mm_{i}})$
%	\end{center}
  Intuitively, each tick in time (\ie\ each $\redd$ step)   is weighted with its probability to take place, which is  $\mu_i\{c \mid c\not\in\Term{\A}\} = 1-\nnorm{\mm_{i}}$ (where $\mu_i$ is the distribution over $A$ associated to $\m_i$). We refer to~\cite{Avanzini}  for the details.
%%%%%%%%%%%%%%%%%%%%

\begin{exa}
	It is immediate to check that in Example~\ref{ex:AST} (Fig.~\ref{fig:AST}), the (unique) maximal {rewrite sequence $ \mathfrak{s} $}  from $[c^1]$ has  $\MeanTime({\mathfrak{s}})=\sum_{n\in \Nat} \frac{1}{2^n} =2$.
\end{exa}

Using this formulation, the following result is immediate.
\begin{cor}\label{RDmeantime1} Let  $\mm\in\MA$.
	$\norm{\nf -}$-RD implies that   all maximal  rewrite sequences from $\mm$ have the \emph{same \MeanTime}.
\end{cor}
{The well-known consequence is that    $\sum_{n} (1-\nnorm{\mm_{n}}) < \infty$  implies $\lim_{n\to \infty} (1- \nnorm{\mm_n})=0$,
hence $\lim_{n\to \infty} \nnorm{\mm_n}=1$.}
  Cor.~\ref{RDmeantime1} means that  if   there  exists
\emph{one}  sequence from $\mm$ with \emph{finite} \MeanTime, all do, hence   $\mm$ is \AST\ and \PAST\@.

\subsection{Analysis of  probabilistic reduction: \emph{Weak} CbV \texorpdfstring{$\lambda$}{lambda}-calculus}\label{sec:weak}

We define  $\PLambda^{\weak}$,
	a  probabilistic analogue of \emph{ weak   Call-by-Value} $\lambda$-calculus (see Section~\ref{background}).
Evaluation is non-deterministic, because in the case of an
application there is no fixed order in the evaluation of the left and right subterms (see Example~\ref{ex:ND}).  We show that
  $ \PLambda^{\weak}$ satisfies $\nf{-}$-RD\@.
%It therefore turns out to be as well behaved as its non-probabilistic counterpart (Corollary~\ref{cor:weak}).
Therefore it has remarkable properties (Cor.~\ref{cor:weak}), analogous  to those of its classical counter-part: the choice of the redex is irrelevant  with respect   to  the \emph{final result}, to its \emph{approximants}, and to the \emph{expected number of steps}.

\subsubsection{The syntax}

The set $\PLambda$ of terms ($M,N,P,Q$) and the set $\Val$ of  values ($V, W$) are defined as follows:
\begin{center}
	{\footnotesize
		$
		\begin{array}{lllc lll}
		M
		&{:}{:}{=} &
		x \midd
		\lambda x.M \midd
		MM \midd M\oplus M	&\quad\quad\quad\quad&
		V
		&{:}{:}{=} &
		x \midd
		\lambda x.M
		\end{array}
		$}
\end{center}
%
%Terms are considered modulo renaming of bound variables.
Free variables are defined as usual. A term $M$ is closed if it has no free variable.
The %(capture-avoiding)
substitution of $V$ for the free occurrences of $x$ in $M$ is denoted $M[x:=V]$.

\paragraph{\pars.} The \pars $(\PLambda,\red)$ is given by the set of terms together with the relation ${\red}\subseteq {\PLambda\times \DSTF{\PLambda}}$ which is
inductively defined by the  rules below.

\begin{center}
	\begin{tabular}{c|c} % chktex 44
%\hline
		$	\begin{array}{l }
		(\lambda x. M)V\red \{M[x:=V]^1\} \\
		P\oplus Q \red \{ P^{1/2}\}+\{Q^{1/2}\}  \\
		\end{array}$
		&
		$ \begin{array}{l }
		\infer[]{MN\red \{(MN_i)^{p_i}\mid i\in I\}}{N\red \{N_i^{p_i}\mid i\in I\}} \quad
		\infer[]{MN\red  \{(M_{i}N)^{p_i} \mid i\in I\}}{M\red  \{M_i^{p_i} \mid i\in I\}}
		\end{array}$\\
	%	\hline
	\end{tabular}
\end{center}

\paragraph{PARS  $\PLambda^{\weak}$.} The calculus  $\PLambda^{\weak}$ is the PARS $(\MDST{\PLambda}, \redd)$, where $\MDST{\PLambda} $ is the set of multidistributions on  $\PLambda$, and % chktex 36
${\redd}  \subseteq {\MDST{\PLambda} \times \MDST{\PLambda}}$ is the lifting (Definition~\ref{def:lift}) of $\red$. % chktex 1

\subsubsection{Examples}

\begin{exa}[Non-deterministic evaluation]\label{ex:ND}A term may have  several reductions. The two reductions here join in one step:
	$\big[P[x:=Q](A\oplus B)\big] \leftleftarrows [((\lam x.P)Q)(A\oplus B)]\rightrightarrows [\two(\lam x.P)Q A, \two (\lam x.P)Q B]$. % chktex 36 chktex 1
\end{exa}

\begin{exa}[Infinitary reduction]\label{ex:infinitary}
	Let $R=(\lam x.xx\oplus\true)(\lam x.xx\oplus\true)$. We have  $[R^1]\tolim\{\true^1\}$.
	{This term models the behaviour we  discussed in  Fig.\ref{fig:AST}}.
\end{exa}

\begin{exa}[Fix-Points]
	$\PLambda^{\weak}$ is expressive enough to allow   fix-point combinators. A simple one is the Turing combinator $\Theta=AA$ where
	$A=\lambda xf.f(xxf)$.
For each  value $F$, $\{\Theta F\} \redd^* \{F(\Theta F)\}$.

\end{exa}

\begin{exa}\label{ex:M}
	The term $PR$ in Example~\ref{ex:motivation2} has  (among others) the following  reduction.
	\begin{align*}
		[PR] & \redd
		[\two P(\true\oplus \false), \two P (\Delta\Delta)] \redd
		[\frac{1}{4} P(\true),  \frac{1}{4} P(\false),\two P (\Delta\Delta)]\\
		& \redd^* [\frac{1}{4}(\true \xor  \true), \frac{1}{4}(\false \xor \false), \frac{1}{4}\Delta\Delta]\redd [ \frac{1}{4}\false, \frac{1}{4}\false, \two \Delta\Delta] \dots
		\end{align*}
			We conclude that  $PR\tolim\{\false^{1/2}\}$.

	\end{exa}

\subsubsection{Properties}
\begin{thm}\label{thm:weak}
	$\PLambda^{\weak}$ satisfies  $\obs$-RD, with $\obs = \nf -$, because it satifies the $\obs$-diamond property.
\end{thm}

\begin{proof}
	We prove  the $\obs$-diamond property, using the definition of lifting and  induction on the structure of the terms (see Appendix~\ref{app:WCBV}).
\end{proof}
Therefore, each $\mm$ satisfies the following properties:

\begin{cor}\label{cor:weak}
		\begin{itemize}
			\item \emph{All}    rewrite sequences from $\mm$ converge to the same limit distribution.
			\item  \emph{All}    rewrite sequences from $\mm$
		have the same  \emph{expected} termination  time \MeanTime.
			\item If $\mm\redd^k \tms$ and $\mm\redd^k\tmt$, then $\nf \tms = \nf \tmt$, $\forall \tms,\tmt,k$.
			%	\item  If $\mm$ has a  sequence $\seq \mm$  such that  $\MeanTime (\seq \mm) < \infty$ then $\mm$ is \PAST\@.
		\end{itemize}
	\end{cor}

\subsection{More diamonds.}\label{sec:more_diamons}
We have discussed weak evaluation of  Call-by-Value $\lam$-calculus, because this is arguably the most relevant paradigm for functional programming.
Similar properties hold for several other reductions from the literature of $\lam$-calculus, we just mention a few which are relevant to the probabilistic setting.

Call-by-Name $\lam$-calculus   has a \emph{non-deterministic variant of head reduction} (well studied in Linear Logic) whose %we write \RED{it} $\hred$ and recall in \Cref{sec:head}---whose
normal forms  are   precisely the head normal forms.
Exactly as weak reduction for Call-by-Value, this variant is well known to satisfy the form of diamond  which gives Random Descent.
Another   well-known calculus with a similar property is  surface reduction in  Simpson's  linear $\lambda$-calculus~\cite{Simpson05}.
For both---head CbN and surface reduction---Random Descent extends to the corresponding probabilistic reductions, which  satisfy similar   properties as those of $\PLambda^\weak$ (the proof is an easy variation of the one given here).
 All three reductions are used in~\cite{FaggianRonchi}.

Another calculus  which satisfy Random Descent  is
Lafont's interaction nets~\cite{Lafont90}---we  expect that its  extension  with a probabilistic choice  would also satisfy Weighted Random Descent.

\section{PARS:\texorpdfstring{\@}{} Comparing Strategies }\label{sec:comparing_pars}
In this section, we briefly examine  the notion of $\B$ in the setting of PARS\@. We focus on the following question:
\begin{center}
	``is there a strategy which is guaranteed to reach a normal form\\ with greatest probability''?
\end{center}

	\paragraph{ARS Normalizing  Strategies.}

The strategy	${\ered}\subseteq {\red}$ is a	 \emph{normalizing strategy}  for $\red$ if
	  whenever  $c\in A$ has a normal form, then
	\emph{every} maximal $\ered$-sequence from $c$ ends in a normal form.

\paragraph{PARS  Normalizing Strategies.} Let $(\MA, \redd)$ be a PARS\@.
We recall that $\pLim(\mm)=\{p \st \mm\tolimp p\}$. {We write  $ q\geq \pLim(\mm)$ if  for each  $p\in \pLim(\mm)$,  $q \geq p$.  Similarly for $ \leq $.}

\begin{defi}\label{def:normalizing} Given a PARS $(\MA, \redd)$,
		a  strategy ${\redds}$ for $\redd$ is \textbf{(asymptotically) normalizing}  if for each $\mm$,
		 \emph{each}  ${\redds}$-sequence starting from $\mm$ converges with the same probability $p_{\max}(\mm) \geq \pLim(\mm)$.
		A strategy ${\reddr}$ for $\redd$ is  \textbf{(asymptotically) perpetual} if for each $\mm$, each  ${\reddr}$ sequence from $\mm$ converges with the same probability   $p_{\min}(\mm) \leq \pLim(\mm)$.
\end{defi}

It is immediate that   $\B(\blue{\redds},\redd)$ with $\obs=\nnorm -$ implies that  $\blue{\redds}$  is normalizing.
By using the results in \refsec{comparing}, we have  a method to prove that a strategy is  normalizing or perpetual by means of a local
condition.
\begin{cor}[Normalizing criterion]\label{cor:normalizing2}\label{LD_norm}
	Let  $\obs$ be $ \nnorm{-}$.  It holds that:
	\begin{enumerate}
		\item  $\LD(\blue{\redds},{\redd})$ implies that  $\blue{\redds}$ is  \emph{asymptotically normalizing}.
		\item  $\LD(\redd,\todo{\reddr})$ implies that  $\todo{\reddr}$ is  \emph{asymptotically perpetual}.
	\end{enumerate}
\end{cor}

\paragraph{Expected Number of Steps.}
Let  $\obs=\nnorm -$. Using Equation~\eqref{eq:ET} in Section~\ref{sec:meantime},
it is easy to check  that if  $\obs$-better$(\blue{\redds},\redd)$ holds,
and $\mathfrak s$ is a $\blue{\redds}$-sequence, then  \MeanTime($\mathfrak s $) $\leq$\MeanTime($\mathfrak t $), for each
$\mathfrak t$  $\redd$-sequence. Notice  that  $\obs$-better$(\blue{\redds},\redd)$ also implies that \MeanTime($\mathfrak s $) is the same  for \emph{any}  $\blue{\redds}$-sequence  $\mathfrak s $.
Therefore,  to establish $\obs$-better$(\blue{\redds},\redd)$ implies  not only    that  the strategy  $\blue{\redds}$  is asymptotically normalizing, but  also that it is   of minimal    \emph{expected termination time}. A similar, dual observation holds for the  perpetuity criterion.

%\paragraph{Greatest Element.} Finally,  we mention that  $\LB$
%provides another method  to establish  \UNlim, and therefore the fact that for each $\mm$, $\sem \mm$ is defined.
%\begin{prop}[Greatest Element] Given $\PP=\{\Pcal, \nf{-}\}$,
%	if there is a strategy $\blue{\redds}$ such that $\nf{-}$-$LB(\blue{\redds}, \redd$)  then $\PP$ satisfies \UNlim.
%\end{prop}

\section{Further work and discussion.}\label{sec:large}
\paragraph{A larger  example of application}
\renewcommand{\Red}{\Rightarrow}
\newcommand{\cbv}{\mathtt{cbv}}
Let us illustrate with an example the use of   the tools which we have developed. We do so by summarizing   further work~\cite{FaggianRonchi}. There,  for each of the following, Plotkin's  Call-by-Value~\cite{PlotkinCbV}, Call-by-Name, and  Simpson's  linear
$\lambda$-calculus~\cite{Simpson05}, a fully fledged probabilistic extension is developed.  Each probabilistic  calculus satisfies confluence, and a form of standardization (surface standardization).
To obtain confluence, only  the probabilistic reduction  is constrained, while  $\beta$ reduction is unrestricted.
In the three calculi, the role of asymptotically  standard strategy is played by  a  reduction which is  \emph{non-deterministic} but \emph{satisfies Random Descent}---this is necessary, because with the (usual) deterministic strategy, a standardization result for finite sequences fails.

Let us see some details.

The notion of result which is studied in~\cite{FaggianRonchi} are, respectively,   values in  CbV, head normal forms in  CbN, and surface normal forms in  the linear calculus.
Once confluence is established,~\cite{FaggianRonchi} relies on the abstract  results given  in \refsec{PARS_UN} to conclude  that---in each calculus---the evaluation of a  program   $\mm$ leads to  a unique maximal result $ \den \mm $---the greatest limit distribution.~\cite{FaggianRonchi}  then studies the question  \emph{``is there a strategy which is guaranteed  to reach  the   unique  result (\emph{asymptotic standardization})?''}. Again,
  key elements   rely  on the abstract  tools developed here;
  let us sketch the construction.

  We focus on  the CbV calculus, namely   $\PLambda^\cbv= (\MDST{\PLambda}, \Rightarrow)$, where
  $ \MDST{\PLambda} $ is as in \refsec{weak}, and $\Rightarrow$ is the (more general) reduction which is defined in~\cite{FaggianRonchi}.
  The role of standard strategy  is played by a
 relation $\Red_s \subseteq \Red$ which is a (more relaxed) lifting of the weak reduction $\red$ defined in \refsec{weak}.
    The construction then  goes as follows.
\begin{enumerate}
	\item  First, it is proved that $\Red_s$ is asymptotically complete for $\Red$. Note however that $\Red_s$ is not guaranteed to  compute $\den \mm$.

 \item It is observed that the relation $\redd$ as defined in \refsec{weak} is asymptotically complete for $\Red_s$, and therefore for $\Red$.

 \item The properties of $\redd$ which are proved in \refsec{weak} guarantee that, from $\mm$,  the limit of \emph{any} $\redd$-sequence is \emph{the same}, and it is exactly $\den \mm$.

\end{enumerate}
 A similar reasoning applies to Call-by-Name.

 Point 3.\ has another implication: for both CbN and CbV,   the leftmost  strategy  reaches the best possible limit distribution (respectively over  \emph{values} and over \emph{head normal forms}). This is  remarkable for two reasons. First---as we already observed in \refsec{motivation}---the leftmost strategy is the  deterministic strategy which has been  adopted in the literature of probabilistic $\lambda$-calculus, in either  its CbV (\cite{KollerMP97,DalLagoMZ11}) or its  CbN version (\cite{DiPierroHW05,EhrhardPT11}), but  without any completeness result with respect to \emph{probabilistic} computation.  The work in~\cite{FaggianRonchi}  offers an ``a-posteriori'' justification for its  use. Second, the result  is non-trivial,
 because  in the probabilistic case, a standardization result for finite sequences  using the leftmost strategy fails for both CbV and CbN.  The tools in  Section~\ref{sec:RD_pars} allow for an elegant solution.

\renewcommand{\Red}{\redd}

\paragraph{ On the \emph{necessity} for non-deterministic evaluation and Random Descent in probabilistic $\lam$-calculi}%
\label{sec:NDE}
A programming language which   is built on a $\lam$-calculus implements a specific evaluation strategy. Typically, evaluation is given by a   strategy $\ered$ of the general reduction $ \red$. In this paper, we studied a property of strategies which is more flexible than determinism,  Random Descent.
Why not simply  fix    a deterministic  strategy?
This choice has several motivations. Non-deterministic evaluation  is a useful feature,   which  supports optimization techniques and parallel/distributed implementation,
but in some cases it is also a \emph{necessity}
and a key  \emph{reasoning tool}---this appears clearly in the probabilistic case.

We illustrate this with two examples from the literature on
probabilistic $\lam$-calculus,~\cite{FaggianRonchi} and~\cite{CurziP20}. Here we discuss the most familiar of all reductions:  Call-by-Name $\lam$-calculus with \emph{head reduction} (similar arguments hold for weak reduction in CbV $ \lam$-calculus).
The usual definition of head reduction~\cite{Barendregt} is deterministic, but it also has a non-deterministic variant (well studied in Linear Logic) whose %we write \RED{it} $\hred$ and recall in \Cref{sec:head}---whose
normal forms  are   precisely the head normal forms. We write this reduction simply $\hred$.
 Exactly as weak reduction for Call-by-Value,  $\hred$  is well known to have Random Descent, and the same hold
for its probabilistic incarnation (an explicit definition is in~\cite{FaggianRonchi}, Ch.~X).

\begin{itemize}
	\item In~\cite{FaggianRonchi}, {moving} from head reduction to its non-deterministic variant $\hred$
allows to obtain a  standardization result, which was known~\cite{Alberti14,Leventis19}  not to hold  when adopting usual, left-to-right head reduction   (see~\cite{FaggianRonchi}, Ex.~45 for a counter-example).

\item Similarly,  %a key  ingredient in
Curzi and Pagani~\cite{CurziP20} %is moving
{move} from {usual} head reduction to head spine reduction, which in turn is included in $\hred$.
The fact that the
evaluation order is not left-to-right, but still
there is no difference
with respect to head normal forms is crucial
to obtain the result of that paper.
\end{itemize}

\section{Conclusions}\label{sec:discussion}\label{sec:conclusion}
The  motivation behind this work is the need for theoretical tools to support the study of operational properties in probabilistic computation,  similarly to the role that ARS have for classical computation.

We have investigated several abstract properties of probabilistic rewriting, and  how the behaviour of different rewrite sequences starting from the same element  compare w.r.t.\ normal forms.
To guarantee that the result of a computation is well defined, we have introduced  and studied  the property \UNlim,  a robust  probabilistic  analogue of the notion of unique normal form. In particular, we have analyzed its relation with (various notions of) confluence.
We also investigated   relations    with normalization (\WNlim) and termination (\SNlim), and between these notions. We have  developed the notions of   \EB\ and $\B$ as tools to analyze and compare PARS strategies.
\EB\ is an alternative to strict determinism, analogous to Random Descent for ARS (non-determinism is irrelevant w.r.t.\ a chosen event of interest). The notion of $\B$
  provides a sufficient criterion to establish that a strategy is \emph{normalizing} (resp. \emph{perpetual}) \ie\ the strategy is guaranteed to lead to a result with maximal (resp.\ minimal) probability.

We have illustrated  our techniques by studying a probabilistic extension of  weak call-by-value $\lambda$-calculus; it has analogous properties to its classical counterpart: all  rewrite  sequences converge to the \emph{same result}, in  the same \emph{expected number of steps}.

\paragraph{One-Step Reduction and Expectations.}
In this paper, we focus on \emph{normal forms} and properties {related to  the event} $\Term{\A}$. However,  we believe that the methods
%{we propose 	are more general, and}
would   allow us to compare strategies w.r.t.\ other properties and  random variables of the system.
The formalism  seems especially well suited to express  the \emph{expected value} of  random variables.
A  key feature of the binary relation $\redd$ is to exactly capture the ARS  notion  of \emph{one-step reduction }
(in contrast   to  \emph{one or no step}), with a  gain which  is two-folded.
\begin{enumerate}
	\item \emph{Probability Theory}. Because all terms in the distribution  are forced to reduce at the same pace, a rewrite sequence  faithfully represents the evolution in time of the system (\ie\, if $\mm\redd^i \mm_i$, then $\mm_i$ captures the  state at time $i$ of all possible  paths $a_0\red \dots \red a_i$).
	This makes the formalism  well suited to express the expected value %(\ie, the probabilistic average)
	of stochastic processes.

	\item \emph{Rewrite Theory}.
		The results in Sections~\ref{sec:RD},~\ref{sec:comparing},~\ref{sec:weak}, crucially rely on \emph{exactly} one-step reduction.
		The reason why this is crucial, is  similar to the classical fact that  termination follows from  normalization by the diamond property [Newman 1942], but not by the  very similar property
		$ b \leftarrow a \rightarrow c \Rightarrow \exists
		d ~( b \rightarrow^= d ~{}^=\hspace{-4pt}\leftarrow c $)
		(see [Terese], 1.3.18).
\end{enumerate}

\paragraph{Finite Approximants.}
\EB\ characterizes the case  when (not only at the limit, but also at the level of the approximants) the non-deterministic choices are irrelevant. %and \Domination\  the case  when a strategy $\S$ is better than another strategy.
The  notion of  approximant which we have studied here is ``stop after a   number  $k$ of steps'' ($k\in \Nat$). We can consider different notion of approximants. For example,
we could also wish to stop the evolution of the system when it reaches a normal form with probability $p$.
Our method can easily be adapted to analyze this case.
We believe it is also possible to extend to the probabilistic setting the results in~\cite{OostromT16}, which would  go further in   this direction.

\paragraph{Further and future work.}
%view belief

In this paper, we have studied  existence and  uniqueness of the result of asymptotic computation.
 The next goal  is to  study how to compute such a result, \ie the study of reduction strategies---this is the object of  current investigation.~\cite{Oostrom07}  makes  a convincing case of
the power of the RD methods for ARS, by using a large range of examples  from
the literature, to elegantly and uniformly revisit normalization results of various $\lambda$-calculi. We cannot do the same here, because the  rich development of strategies for $\lambda$-calculus  has not yet an analogue in the probabilistic case. Nevertheless, we hope that the availability of tools to  analyze  PARS strategies will contribute to their development.

%\pink{\paragraph{Asymptotic rewriting.}  XXXXX}

\paragraph{Acknowledgements.}
{This work   benefited of fruitful discussions with U. Dal Lago,   T. Leventis, and B. Valiron. I  am very grateful to   V. Van Oostrom  for valuable comments and  suggestions.} The proof of Prop.~\ref{prop:bool_UN} appearing here is a simplification of the original  one,  thanks to the insightful remarks of an anonymous reviewer.

%%%%%%%%%%%%%%%

\bibliographystyle{alphaurl}

	\bibliography{biblioPARSshort}

\begin{thebibliography}{KKSdV95}

\bibitem[AB02]{AriolaBlom02}
Zena~M. Ariola and Stefan Blom.
\newblock Skew confluence and the lambda calculus with letrec.
\newblock {\em Annals of Pure and Applied Logic}, 117(1):95 -- 168, 2002.

\bibitem[AC98]{AmadioCurien}
Roberto~M. Amadio and Pierre{-}Louis Curien.
\newblock {\em Domains and lambda-calculi}, volume~46 of {\em Cambridge tracts
  in theoretical computer science}.
\newblock Cambridge University Press, 1998.

\bibitem[ACN18]{AgrawalC018}
Sheshansh Agrawal, Krishnendu Chatterjee, and Petr Novotn{\'{y}}.
\newblock Lexicographic ranking supermartingales: an efficient approach to
  termination of probabilistic programs.
\newblock {\em {PACMPL}}, 2({POPL}):34:1--34:32, 2018.

\bibitem[Alb14]{Alberti14}
Michele Alberti.
\newblock {\em On operational properties of quantitative extensions of
  {\(\lambda\)}-calculus}.
\newblock PhD thesis, Aix-Marseille University, France, 2014.

\bibitem[ALY20]{Avanzini}
Martin Avanzini, Ugo~Dal Lago, and Akihisa Yamada.
\newblock On probabilistic term rewriting.
\newblock {\em Sci. Comput. Program.}, 185, 2020.

\bibitem[AMS06]{AghaMS06}
Gul~A. Agha, Jos{\'{e}} Meseguer, and Koushik Sen.
\newblock {PM}aude: Rewrite-based specification language for probabilistic
  object systems.
\newblock {\em Electr. Notes Theor. Comput. Sci.}, 153(2):213--239, 2006.

\bibitem[Bar84]{Barendregt}
Hendrik~Pieter Barendregt.
\newblock {\em The Lambda Calculus -- Its Syntax and Semantics}, volume 103 of
  {\em Studies in logic and the foundations of mathematics}.
\newblock North-Holland, 1984.

\bibitem[BG06]{BournezG05}
Olivier Bournez and Florent Garnier.
\newblock Proving positive almost sure termination under strategies.
\newblock In {\em Rewriting Techniques and Applications, {RTA}}, pages
  357--371, 2006.

\bibitem[BK02]{BournezK02}
Olivier Bournez and Claude Kirchner.
\newblock Probabilistic rewrite strategies. {A}pplications to {ELAN}.
\newblock In {\em Rewriting Techniques and Applications, {RTA}}, pages
  252--266, 2002.

\bibitem[CH98]{Hindley}
N.~Cagman and J.R. Hindley.
\newblock Combinatory weak reduction in lambda calculus.
\newblock {\em Theor. Comput. Sci.}, 1998.

\bibitem[CP20]{CurziP20}
Gianluca Curzi and Michele Pagani.
\newblock The benefit of being non-lazy in probabilistic
  {\(\lambda\)}-calculus: Applicative bisimulation is fully abstract for
  non-lazy probabilistic call-by-name.
\newblock In {\em {LICS} '20: 35th Annual {ACM/IEEE} Symposium on Logic in
  Computer Science, 2020}, pages 327--340. {ACM}, 2020.

\bibitem[DAGG11]{DiazArrighiGG}
Alejandro D{\'{\i}}az{-}Caro, Pablo Arrighi, Manuel Gadella, and Jonathan
  Grattage.
\newblock Measurements and confluence in quantum lambda calculi with explicit
  qubits.
\newblock {\em Electr. Notes Theor. Comput. Sci.}, 270(1):59--74, 2011.

\bibitem[DKP91]{DershowitzKP91}
Nachum Dershowitz, St{\'{e}}phane Kaplan, and David~A. Plaisted.
\newblock Rewrite, rewrite, rewrite, rewrite, rewrite, ...
\newblock {\em Theor. Comput. Sci.}, 83(1):71--96, 1991.

\bibitem[DLM08]{LagoM08}
Ugo Dal~Lago and Simone Martini.
\newblock The weak lambda calculus as a reasonable machine.
\newblock {\em Theor. Comput. Sci.}, 398(1-3):32--50, 2008.

\bibitem[DLMZ11]{DalLagoMZ11}
Ugo Dal~Lago, Andrea Masini, and Margherita Zorzi.
\newblock Confluence results for a quantum lambda calculus with measurements.
\newblock {\em Electr. Notes Theor. Comput. Sci.}, 270(2):251--261, 2011.

\bibitem[DLZ12]{LagoZ12}
Ugo Dal~Lago and Margherita Zorzi.
\newblock Probabilistic operational semantics for the lambda calculus.
\newblock {\em {RAIRO} - Theor. Inf. and Applic.}, 46(3):413--450, 2012.

\bibitem[DM18]{DiazMartinez17}
Alejandro D{\'{\i}}az{-}Caro and Guido Martinez.
\newblock Confluence in probabilistic rewriting.
\newblock {\em Electr. Notes Theor. Comput. Sci.}, 338:115--131, 2018.

\bibitem[dP95]{deLiguoroP95}
Ugo de{'}Liguoro and Adolfo Piperno.
\newblock Non deterministic extensions of untyped lambda-calculus.
\newblock {\em Inf. Comput.}, 122(2):149--177, 1995.

\bibitem[DPHW05]{DiPierroHW05}
Alessandra Di~Pierro, Chris Hankin, and Herbert Wiklicky.
\newblock Probabilistic lambda-calculus and quantitative program analysis.
\newblock {\em J. Log. Comput.}, 15(2):159--179, 2005.

\bibitem[EPT11]{EhrhardPT11}
Thomas Ehrhard, Michele Pagani, and Christine Tasson.
\newblock The computational meaning of probabilistic coherence spaces.
\newblock In {\em Proceedings of the 26th Annual {IEEE} Symposium on Logic in
  Computer Science, {LICS} 2011}, pages 87--96. {IEEE} Computer Society, 2011.

\bibitem[Fag19]{Faggian19}
Claudia Faggian.
\newblock Probabilistic rewriting: Normalization, termination, and unique
  normal forms.
\newblock In {\em 4th International Conference on Formal Structures for
  Computation and Deduction, {FSCD} 2019}, volume 131 of {\em LIPIcs}, pages
  19:1--19:25. Schloss Dagstuhl, 2019.

\bibitem[FC19]{FuC19}
Hongfei Fu and Krishnendu Chatterjee.
\newblock Termination of nondeterministic probabilistic programs.
\newblock In {\em Verification, Model Checking, and Abstract Interpretation
  {VMCAI}}, pages 468--490, 2019.

\bibitem[FH15]{FioritiH15}
Luis Mar{\'{\i}}a~Ferrer Fioriti and Holger Hermanns.
\newblock Probabilistic termination: Soundness, completeness, and
  compositionality.
\newblock In {\em {POPL}}, pages 489--501, 2015.

\bibitem[FR19]{FaggianRonchi}
Claudia Faggian and Simona Ronchi~Della Rocca.
\newblock Lambda calculus and probabilistic computation.
\newblock In {\em 34th Annual {ACM/IEEE} Symposium on Logic in Computer
  Science, {LICS} 2019}, pages 1--13. {IEEE}, 2019.

\bibitem[HFCG19]{Huang0CG19}
Mingzhang Huang, Hongfei Fu, Krishnendu Chatterjee, and Amir~Kafshdar
  Goharshady.
\newblock Modular verification for almost-sure termination of probabilistic
  programs.
\newblock {\em Proc. {ACM} Program. Lang.}, 3({OOPSLA}):129:1--129:29, 2019.

\bibitem[How70]{Howard}
W.A. Howard.
\newblock Assignment of ordinals to terms for primitive recursive functionals
  of finite type.
\newblock In {\em Intuitionism and Proof Theory}, 1970.

\bibitem[KC17]{Kirkeby}
Maja~H. Kirkeby and Henning Christiansen.
\newblock Confluence and convergence in probabilistically terminating reduction
  systems.
\newblock In {\em Logic-Based Program Synthesis and Transformation - 27th
  International Symposium, {LOPSTR} 2017}, pages 164--179, 2017.

\bibitem[KdV05]{KlopV05}
Jan~Willem Klop and Roel~C. de~Vrijer.
\newblock Infinitary normalization.
\newblock In {\em We Will Show Them! Essays in Honour of Dov Gabbay, Volume
  Two}, pages 169--192, 2005.

\bibitem[Ken92]{Kennaway92}
Richard Kennaway.
\newblock On transfinite abstract reduction systems.
\newblock Tech. rep., CWI, Amsterdam, 1992.

\bibitem[KKMO18]{KaminskiKMO18}
{Benjamin Lucien} Kaminski, Joost{-}Pieter Katoen, Christoph Matheja, and
  Federico Olmedo.
\newblock Weakest precondition reasoning for expected runtimes of randomized
  algorithms.
\newblock {\em J. {ACM}}, 65(5):30:1--30:68, 2018.

\bibitem[KKSdV95]{KennawayKSV95}
Richard Kennaway, Jan~Willem Klop, M.~Ronan Sleep, and Fer{-}Jan de~Vries.
\newblock Transfinite reductions in orthogonal term rewriting systems.
\newblock {\em Inf. Comput.}, 119(1):18--38, 1995.

\bibitem[KMP97]{KollerMP97}
Daphne Koller, David~A. McAllester, and Avi Pfeffer.
\newblock Effective bayesian inference for stochastic programs.
\newblock In {\em National Conference on Artificial Intelligence and Innovative
  Applications of Artificial Intelligence Conference, {AAAI} 97, {IAAI} 97},
  pages 740--747, 1997.

\bibitem[Laf90]{Lafont90}
Yves Lafont.
\newblock Interaction nets.
\newblock In {\em Conference Record of the Seventeenth Annual {ACM} Symposium
  on Principles of Programming Languages, San Francisco, California, USA,
  January 1990}, pages 95--108. {ACM} Press, 1990.

\bibitem[Ler90]{Leroy-ZINC}
Xavier Leroy.
\newblock The {ZINC} experiment: an economical implementation of the {ML}
  language.
\newblock Technical report 117, INRIA, 1990.
\newblock URL: \url{http://gallium.inria.fr/~xleroy/publi/ZINC.pdf}.

\bibitem[L{\'{e}}v78]{LevyPhD}
Jean{-}Jacques L{\'{e}}vy.
\newblock {\em R\'{e}ductions Corrects et Optimales dans le Lambda-Calcul}.
\newblock PhD thesis, Universit\'{e} Paris VII, 1978.

\bibitem[Lev19]{Leventis19}
Thomas Leventis.
\newblock A deterministic rewrite system for the probabilistic
  $\lambda$-calculus.
\newblock {\em Mathematical Structures in Computer Science},
  29(10):1479–1512, 2019.

\bibitem[LFR21]{LagoFR21}
Ugo~Dal Lago, Claudia Faggian, and Simona Ronchi~Della Rocca.
\newblock Intersection types and (positive) almost-sure termination.
\newblock {\em Proc. {ACM} Program. Lang.}, 5({POPL}):1--32, 2021.

\bibitem[LFVY17]{popl}
Ugo~Dal Lago, Claudia Faggian, Beno{\^{\i}}t Valiron, and Akira Yoshimizu.
\newblock The geometry of parallelism: classical, probabilistic, and quantum
  effects.
\newblock In {\em Proceedings of the 44th {ACM} {SIGPLAN} Symposium on
  Principles of Programming Languages, {POPL} 2017, Paris}, pages 833--845.
  {ACM}, 2017.

\bibitem[Mar13]{Marlow}
Simon Marlow.
\newblock {\em Parallel and Concurrent Programming in Haskell}.
\newblock O’Reilly Media, 2013.

\bibitem[MMKK18]{McIverMKK18}
Annabelle McIver, Carroll Morgan, {Benjamin Lucien} Kaminski, and
  Joost{-}Pieter Katoen.
\newblock A new proof rule for almost-sure termination.
\newblock {\em {PACMPL}}, 2({POPL}):33:1--33:28, 2018.

\bibitem[New42]{Newman}
Mark Newman.
\newblock On theories with a combinatorial definition of {“Equivalence”}.
\newblock {\em Annals of Mathematics}, 43(2):223–243, 1942.

\bibitem[Plo75]{PlotkinCbV}
Gordon~D. Plotkin.
\newblock Call-by-name, call-by-value and the lambda-calculus.
\newblock {\em Theor. Comput. Sci.}, 1(2):125--159, 1975.

\bibitem[PPT05]{ParkPT05}
Sungwoo Park, Frank Pfenning, and Sebastian Thrun.
\newblock A probabilistic language based upon sampling functions.
\newblock In {\em Proceedings of the 32nd {ACM} {SIGPLAN-SIGACT} Symposium on
  Principles of Programming Languages, {POPL} 2005}, pages 171--182. {ACM},
  2005.

\bibitem[Put94]{Puterman94}
Martin~L. Puterman.
\newblock {\em Markov Decision Processes: Discrete Stochastic Dynamic
  Programming}.
\newblock John Wiley \& Sons, Inc., New York, NY, USA, 1st edition, 1994.

\bibitem[Rab63]{Rabin63}
Michael~O. Rabin.
\newblock Probabilistic automata.
\newblock {\em Information and Control}, 6(3):230--245, 1963.

\bibitem[RP02]{RamseyP02}
Norman Ramsey and Avi Pfeffer.
\newblock Stochastic lambda calculus and monads of probability distributions.
\newblock In {\em Conference Record of {POPL} 2002: The 29th {SIGPLAN-SIGACT}
  Symposium on Principles of Programming Languages, 2002}, pages 154--165.
  {ACM}, 2002.

\bibitem[Sah78]{Saheb-Djahromi78}
N.~Saheb{-}Djahromi.
\newblock Probabilistic {LCF}.
\newblock In {\em Mathematical Foundations of Computer Science}, pages
  442--451, 1978.

\bibitem[San71]{Santos69}
Eugene~S. Santos.
\newblock Computability by probabilistic turing machines.
\newblock In {\em Transactions of the American Mathematical Society}, pages
  159:165--184, 1971.

\bibitem[Sim05]{Simpson05}
Alex~K. Simpson.
\newblock Reduction in a linear lambda-calculus with applications to
  operational semantics.
\newblock In {\em Rewriting Techniques and Applications, {RTA}}, pages
  219--234, 2005.

\bibitem[Ter03]{Terese03}
Terese.
\newblock {\em Term Rewriting Systems}, volume~55 of {\em Cambridge Tracts in
  Theoretical Computer Science}.
\newblock Cambridge University Press, 2003.

\bibitem[vO07]{Oostrom07}
Vincent van Oostrom.
\newblock Random descent.
\newblock In {\em Term Rewriting and Applications, {RTA}}, page 314–328,
  2007.

\bibitem[vOT16]{OostromT16}
Vincent van Oostrom and Yoshihito Toyama.
\newblock Normalisation by random descent.
\newblock In {\em Formal Structures for Computation and Deduction, {FSCD}},
  pages 32:1--32:18, 2016.

\end{thebibliography}

		\appendix

\section{Omitted Proofs and Further Details}

\subsection{ Section~\ref{sec:confluenceQ}. Confluence and \texorpdfstring{\UNlim}{UNlim}}\label{app:ARSI}
 Note that for ARS, \UN\ does not imply confluence. Similarly, for QARS and PARS, \UNlim\ does not imply Confluence or Skew-Confluence.
{\begin{exa}[\UNlim\ does not imply Confluence or Skew-Confluence]\label{ex:UN}
	Consider  the PARS generated by the following \pars:
	\begin{center}
		$c\to\{a^{\two}, \false^\two\}$, $c\to\{\true^{\two},b^\two\}$, $a\to \{\true^\two, a^\two\}$, $b\to \{\false^\two, b^\two\}$.
	\end{center}
	Each PARS element $[c],[a],[b], \dots$ has a unique limit. No version of confluence holds, as it is easily seen taking the two sequences 	$[c]\redd [\two a, \two \false]$ and $[c]\redd [\two b, \two \true]$, and observing that they do not join, because there exists no $\mm$ such that  $ [\two a, \two \false]\redd^* \mm  $  and
	$[\two b, \two \true] \redd^* \mm $
\end{exa}}

\begin{rem}[QARS limits vs  ARSI infinite normal forms] The notion of limit which is associated to  QARS is more general than the notion
	of infinite normal form which is defined for ARSI~\cite{AriolaBlom02}.
	%the notion of limit is more general---infinite normal forms of~\cite{AriolaBlom02} are only a special case.
	Note that
	in  the setting of~\cite{AriolaBlom02} the following holds (Theorem 5.4 there):
	\begin{center}
		{``an ARSI is skew confluent \emph{if and only if}  it has  unique infinite normal forms''}.
	\end{center} An analogue property \emph{does not hold} for QARS\@. {Even for PARS, the ``if'' direction fails (\refex{UN}).}

\end{rem}

\subsection{Section~\ref{sec:weak}. Weak CbV \texorpdfstring{$\lambda$}{lambda}-calculus}\label{app:WCBV}

\begin{theorem*}[\ref{thm:weak}] $\PLambda^{\weak}$ satisfies the $\obs$-diamond property, with $\obs = \nf -$.
\end{theorem*}
%\begin{thm}
%	Weak $\Lambda_{\oplus}$ satisfies  $\Enf$-Balance.
%\end{thm}

	\begin{proof} We show by induction on the structure of the term $M$ that  for all  pairs of one-step reductions $\tmt \leftleftarrows [M^1] \redd \tms$, either $\tmt = \tms$,    (and therefore   $\exists \tmu. ~\tmt \redd \tmu \leftleftarrows \tms$) or the following hold: (1.)  $\nf \tms =\zero= \nf{\tmt}$ (\ie they are equal because they both take value $0$ everywhere), and     (2.) exists $\tmu$ such that $\tmt \redd \tmu \leftleftarrows \tms$.

			\begin{itemize}
				\item Case $M=x$ or $M=\lam x.P$: no reduction is possible.
				\item Case $M= P\oplus Q$, or $M=(\lam x.N) V$:    only one reduction is possible, and $\tmt=\tms$.
				\item Otherwise,  $M=PQ$, and two cases are possible.
				\begin{itemize}
					\item Assume that both $P$ and $Q$ reduce; $PQ$ has the following reductions:
					\begin{center}
						$\infer{PQ\red  \{(P_iQ)^{p_i}\mid i\in I\}}{P\red \{(P_i)^{p_i}\mid i\in I\}}$ \quad  and \quad $\infer{PQ\red   \{(PQ_j)^{q_j}\mid j\in J\}}{Q\red  \{Q_j^{q_j}\mid j\in J\}}$
					\end{center}
					Observe that none of the $P_iQ$ or  $PQ_j$  is a normal form, hence (1.) holds.	By the definition of  $\red$, the following  holds
\begin{center}
						$\left( \infer{P_iQ \red \{(P_iQ_j)^{q_j}\mid j\in J\}}{Q\red \{Q_j^{q_j}\mid j\in J\}} \right)_{\iI}$
\end{center}
						and therefore by Lifting we have
					$\sum_i p_i\cdot [P_iQ] \redd \sum_i p_i\cdot (\sum_j q_j\cdot [P_iQ_j])  =   \sum_{i,j} p_i q_j \cdot [P_i Q_j]$.
					Similarly we obtain $  \sum_j q_j\cdot [PQ_j]\redd \sum_{i,j} p_i q_j \cdot [P_i Q_j]$.

					\item If one subterm has two  reductions, we conclude by \ih.

					 Let assume that  $P$ has two different redexes (the  case of $Q$ is similar):
					\begin{center}
						$[P]\redd \tms= \sum_i s_i \cdot [S_i]$ and $[P]\redd \tmt =\sum_j t_j\cdot  [T_j]$
					\end{center}
					By induction  hypothesis, two facts hold:  (1.) $\nf\tms =\zero= \nf \tmt $, therefore no $S_i$ and no $T_j$  is a normal form;  (2.) there exist steps $[S_i]\Red \sum_k r_k \cdot[R_{ik}]$  and $[T_j] \Red r_h\cdot[R_{jh}]$
					such that $P\redd \sum_i[S_i] \redd\sum_i(\sum_k s_i r_k\cdot[R_{ik}])   = \tmr'$ and $P\redd \sum_j[T_j] \redd \sum_j(\sum_h t_j r_h\cdot[R_{jh}])=\tmr''$,
					and $\tmr'=\tmr''$.

					For $PQ$ we have
					\begin{center}
						$\infer{PQ\red  \{(S_iQ)^{s_i}\st i\in I\}}{P\red   \{S_i^{s_i}\st i\in I\} }$\quad
						and \quad  $\infer{PQ\red  \{(T_jQ)^{ t_j }\sth j\in J\}}{P\red  \{T_j^{ t_j }\sth j\in J\}}$
					\end{center}
					First, we observe that no $S_iQ$ and no $T_jQ$ is a normal form, hence property (1.) is verified. Moreover,
					it holds that

					$PQ\redd \sum_i s_i\cdot[S_iQ] \redd\sum_i(\sum_k s_i r_k\cdot[R_{ik}Q]) =\tmu'$ and $PQ\redd \sum_j t_j\cdot [T_jQ] \redd \sum_j(\sum_h t_j r_h\cdot[R_{jh}Q])=\tmu''$.
			From $\tmr'=\tmr''$ it follows  that
					$\tmu'=\tmu''$; hence  property (2.) is also verified.
                    \qedhere
				\end{itemize}
			\end{itemize}

	\end{proof}

\end{document}